\documentclass[10pt, conference, letterpaper]{IEEEtran}
\pdfoutput=1
\usepackage{amsmath}
\usepackage{amsthm}
\usepackage{amssymb}
\usepackage{ifthen}
\usepackage{graphicx}
\usepackage{subfigure}
\usepackage{color}
\usepackage{balance}
\usepackage{cite}
\usepackage{url}
\newtheorem{definition}{Definition}

\newtheorem{lemma}{Lemma}

\newtheorem{theorem}{Theorem}
\newcommand{\version}{short}
\newcommand{\cond}[2]{\ifthenelse{\equal{\version}{#1}}{#2}{}}
\IEEEoverridecommandlockouts
\begin{document}
\title{Age- and Deviation-of-Information of Time-Triggered and Event-Triggered Systems}
\author{\IEEEauthorblockN{Mahsa Noroozi, Markus Fidler} %\thanks{}}
\IEEEauthorblockA{Institute of Communications Technology\\ Leibniz Universit\"{a}t Hannover}
}
\maketitle
\thispagestyle{plain} % forces page numbers
\pagestyle{plain}
\begin{abstract}
Age-of-information is a metric that quantifies the freshness of information obtained by sampling a remote sensor. In signal-agnostic sampling, sensor updates are triggered at certain times without being conditioned on the actual sensor signal. Optimal update policies have been researched and it is accepted that periodic updates achieve smaller age-of-information than random updates. We contribute a study of a signal-aware policy, where updates are triggered by a random sensor event. By definition, this implies random updates and as a consequence inferior age-of-information. Considering a notion of deviation-of-information as a signal-aware metric, our results show, however, that event-triggered systems can perform equally well as time-triggered systems while causing smaller mean network utilization.
\end{abstract}
%
%------------------------------------------------------------------------
%
\section{Introduction}
\label{sec:introduction}
We consider a system where a remote sensor is sampled and the samples are transmitted via a network to a monitor. A model of the system is shown in Fig.~\ref{fig:system}. The signal $C(t)$ generated by the sensor changes randomly over time $t$ and the $n$th sample is taken and sent to the network at time $A(n)$. We investigate two different sampling policies. In a time-triggered system, the sampling process is agnostic to the signal and samples are taken after a certain amount of time has elapsed. In an event-triggered system, the sampler is signal-aware and whenever the signal change with respect to the last sample exceeds a threshold, a new sample is generated. Sample $n$ has network service requirement $S_i(n)$ at queue $i$ and it departs from the network to the monitor at time $D(n)$. The monitor does not have a priori knowledge of the distribution and parameters of the sensor signal $C(t)$. Hence, it relies only on the most recent update received, i.e., at time $t$ sample $n^* = \max\{ n : D(n) < t \}$ provides the sensor reading $C(A(n^*))$ generated at time $A(n^*)$.

A key performance metric of such systems is the age-of-information (AoI) that quantifies the freshness of information at the monitor. The AoI is defined as $\Delta(t) = t - A(n^*)$. An example of the progression of the AoI over time is shown in Fig.~\ref{fig:aoisimple}~\cite{kaul:ageofinformationvehicular}. The information of sample $n$ generated at time $A(n)$ ages with slope one with $t$. The monitor selects the most recent sample $n^*$ that it has received. This leads to the linear increase of $\Delta(t)$ with discontinuities whenever a fresher sample becomes available at the monitor and the AoI is reset to the network delay.

The notion of AoI has been introduced in vehicular networks~\cite{kaul:ageofinformationvehicular, kaul:ageofinformationqueue, zinchenko:informationfreshness, tchouankem:messagelifetime}. It has emerged as a very active area of research, being of general importance for a variety of applications in the areas of cyber-physical systems and the Internet of Things. There, particular challenges arise in networked feedback control systems~\cite{champati:ageofinformationfeedbackcontrol, ayan:valueofinformation, klugel:aoipenalty}. Recent surveys are~\cite{yates:ageofinformationsurvey, kosta:ageofinformation}.

\begin{figure}
\begin{center}
\includegraphics[width=0.95\linewidth]{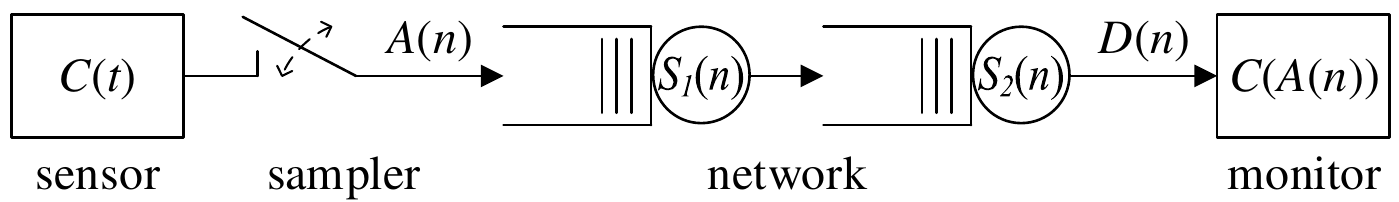}
\caption{System model. At time $A(n)$ the $n$th sample of the sensor signal $C(t)$ arrives at a network of queues, with service times $S_i(n)$. At time $D(n)$ the sample departs from the network to a monitor, conveying the signal $C(A(n))$.}
\label{fig:system}
\end{center}
\end{figure}

A general objective of AoI research is to find update policies that minimize the AoI. Common policies are periodic sampling, random sampling, and zero-wait sampling~\cite{kaul:ageofinformationqueue, yates:lazyistimely, shroff:updateorwait}. The effects of periodic and random sampling on the AoI have been studied in-depth using models of D$\mid$M$\mid$1 and M$\mid$M$\mid$1 queues and variants thereof~\cite{kaul:ageofinformationqueue, inoue:aoisingleserverqueues, champati:ageofinformationgigiqueue, modiano:informationfreshness}, and it is universally accepted that periodic sampling outperforms exponential, random sampling. Zero-wait sampling uses $A(n+1) = D(n)$ for all $n \ge 1$, i.e., reception of sample $n$ by the monitor triggers generation of sample $n+1$. This avoids queueing in the network entirely and achieves good but not necessarily optimal AoI~\cite{yates:lazyistimely, shroff:updateorwait}. Zero-wait sampling differs, however, from our system in Fig.~\ref{fig:system} as it requires feedback of network state information.

Different from these signal-agnostic policies, we consider a signal-aware policy~\cite{yates:ageofinformationsurvey, sun:remoteestimation, ornee:remoteestimation}, where samples are generated in case of a defined, random sensor event. At first sight, this brings about random updates, which may be assumed to have worse AoI performance than time-triggered, periodic updates. Noticing that AoI is a signal-agnostic metric, this may not be unexpected. We define a deviation-of-information (DoI) metric $\Phi(t) = C(t) - C(A(n^*))$ that matches the definition of AoI $\Delta(t) = t - A(n^*)$, but replaces age by the actual deviation of the monitor's signal estimate from the sensor signal $C(t)$.

We employ a max-plus queueing model and stochastic methods of the network calculus to derive bounds of tail delays~\cite{chang:performanceguarantees, ciucu:networkservicecurvescaling2, jiang:stochasticnetworkcalculus, fidler:netcalcguide}. We contribute solutions for AoI and DoI of time- and event-triggered systems. Simulation results that confirm the tail decay rates of our analytical bounds are included. Our results enable finding update rates that minimize the AoI or DoI, respectively. Interestingly, the optimal update rate may differ with respect to the goal of AoI or DoI minimization. While the event-triggered system has larger AoI, our evaluation shows that it requires a lower average update rate to achieve DoI performance similar to the time-triggered system.

The remainder of this work is structured as follows. In Sec.~\ref{sec:relatedwork} we give an overview of related works. Our basic model of a system that is triggered by sensor events is developed in Sec.~\ref{sec:sensormodel} where we also define suitable performance metrics. In Sec.~\ref{sec:aoi} we derive a lemma that is essential for our investigation of DoI. As an immediate corollary this lemma provides tail bounds of delay and AoI of time-triggered and event-triggered systems. We obtain our main result for the DoI in Sec.~\ref{sec:doi}. Brief conclusions are presented in Sec.~\ref{sec:conclusion}.

\begin{figure}
\begin{center}
\includegraphics[width=0.9\linewidth]{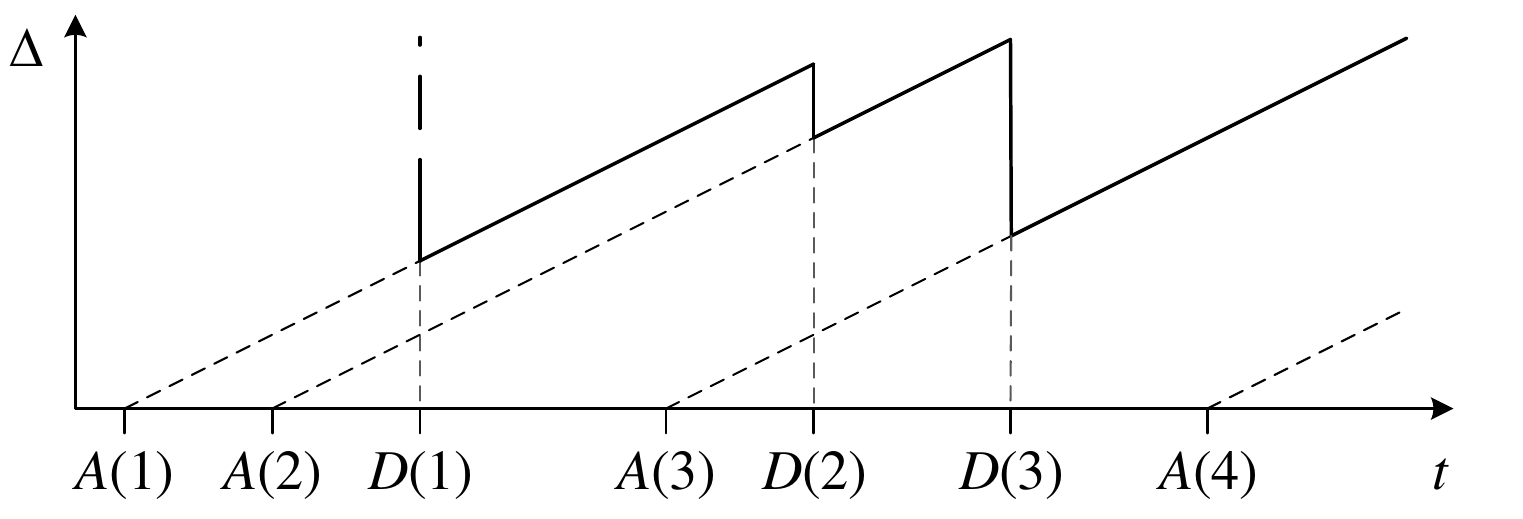}
\caption{Progression of the age-of-information $\Delta(t)$ over time $t$. $A(n)$ and $D(n)$ denote the network arrival and departure time stamps of sample $n$.}
\label{fig:aoisimple}
\end{center}
\end{figure}
\section{Related Work}
\label{sec:relatedwork}
The notion of AoI as a performance metric and its relevance to a wide range of systems have attracted significant research. During the past decade, AoI results of a catalogue of queueing systems have been accomplished~\cite{inoue:aoisingleserverqueues, yates:ageofinformationsurvey, kosta:ageofinformation}. Commonly, the time-average of the AoI, that can be visualized by the area under the curve in Fig.~\ref{fig:aoisimple}, is derived. Further, the peak AoI~\cite{modiano:ageofinformationqueueing, modiano:informationfreshness}, that is the maximal AoI observed immediately before an update is received, and the tail distribution of the AoI~\cite{pappas:ageofinformationnetworkcalculus, champati:ageofinformationmaxplus, champati:ageofinformationdgqueue, noroozi:minplusaoi, rizk:palmaoi} have been studied. In this work, we consider the peak AoI and like~\cite{pappas:ageofinformationnetworkcalculus, champati:ageofinformationmaxplus, noroozi:minplusaoi} we employ techniques from the stochastic network calculus~\cite{chang:performanceguarantees, ciucu:networkservicecurvescaling2, jiang:stochasticnetworkcalculus, fidler:netcalcguide} to estimate tail probabilities.

The starting basis of our work are a number of studies that compare the impact of periodic versus exponential sampling on the AoI. Optimal update rates that minimize the average AoI are considered in~\cite{kaul:ageofinformationqueue} for M$\mid$M$\mid$1, D$\mid$M$\mid$1, and M$\mid$D$\mid$1 queues. It is observed that the random arrivals of the M$\mid$M$\mid$1 queue lead to a 50\% increase of the AoI compared to the D$\mid$M$\mid$1 queue. For last-come first-served queues with and without preemption~\cite{inoue:aoisingleserverqueues} reports accordingly that the AoI of the D$\mid$M$\mid$1 queue outperforms the M$\mid$M$\mid$1 queue. The AoI of GI$\mid$GI$\mid$1$\mid$1 and GI$\mid$GI$\mid$1$\mid$2* queues is investigated in~\cite{champati:ageofinformationgigiqueue} and results are presented for deterministic arrivals and deterministic service, respectively. A comparison of periodic arrivals and Bernoulli arrivals in wireless networks~\cite{modiano:informationfreshness} shows that periodic arrivals outperform Bernoulli arrivals considering average AoI and peak AoI. These results indicate that random sampling may in general perform worse than periodic sampling. A plausible implication is that event-triggered systems may be inferior to time-triggered systems.

While the AoI of a sample increases linearly with time, the actual validity period of that sample depends on the future progression of the sensor signal. Taking this aspect into account appears essential for evaluation of event-triggered systems. A number of works employ a non-linear aging function to represent the value-of-information over time, see the survey~\cite{yates:ageofinformationsurvey}. The evolution of a random sensor signal can, however, not be modeled by a deterministic function.

Sampling governed by an external random process is considered in energy-harvesting systems, where random energy arrivals trigger sensor updates, see~\cite{yates:ageofinformationsurvey} for an overview. Different from these works, the event-triggered systems that we consider are signal-aware, i.e., the progression of the signal itself triggers sensor updates.

More closely related to our work are a number of studies on remote estimation of the state of a linear plant with Gaussian disturbance via a network~\cite{ayan:valueofinformation, champati:ageofinformationfeedbackcontrol, klugel:aoipenalty}. In~\cite{champati:ageofinformationfeedbackcontrol} geometric transmission times with success probability $p$ are assumed, whereas~\cite{klugel:aoipenalty} considers an erasure channel with loss probability $1-p$ and unit service time, and~\cite{ayan:valueofinformation} investigates scheduling for a cellular network. The common target is to minimize the mean-square norm of the state error at the monitor. It is shown that this can be expressed by a non-decreasing function of the AoI, referred to as age-penalty function in~\cite{klugel:aoipenalty} and expressed as value-of-information in~\cite{ayan:valueofinformation}. The result is an equivalent AoI minimization problem~\cite{champati:ageofinformationfeedbackcontrol, klugel:aoipenalty} that is signal-agnostic. AoI minimization is studied in~\cite{kaul:ageofinformationqueue, yates:lazyistimely, shroff:updateorwait}.

Remote estimation of Wiener processes using signal-aware sampling is analyzed in~\cite{sun:remoteestimation} and generalized to Ornstein-Uhlenbeck processes in~\cite{ornee:remoteestimation}. Samples are generated whenever the instantaneous estimation error exceeds a threshold. The policy is proven to minimize the time-average mean-square error of the estimate. For signal-agnostic sampling it is shown that the problem can be recast as AoI minimization. Generally, the policies that are investigated include an adapted zero-wait condition, where a new sample is generated only after the previous sample is delivered, i.e., $A(n+1) \ge D(n)$ for all $n \ge 1$. This avoids the problem of waiting times in network queues but requires feedback information that is not included in our system model, see Fig.~\ref{fig:system}.
%
%------------------------------------------------------------------------
%
\section{Sensor Model and Performance Metrics}
\label{sec:sensormodel}
We model the sensor signal as a random process and define the performance metrics peak AoI and DoI at the monitor.
%
%------------------------------------------------------------------------
%
\subsection{Sensor Model}
We consider a sensor that detects the occurrence of defined, random events indexed $n \in \mathbb{N}$ in order. Time $t \in \mathbb{R}_{0+}$ is continuous and non-negative. We denote $E(n)$ the time of occurrence of event $n \ge 1$, and define $E(0) = 0$. For all $n \ge 1$ it holds that $E(n) \ge E(n-1)$ and $I(n) = E(n) - E(n-1)$ are the inter-event times. The event count
\begin{equation}
C(t) = \max \{n \ge 0: E(n) \le t \},
\label{eq:eventcount}
\end{equation}
denotes the cumulative number of events that occurred in $(0,t]$. By definition $C(t) \in \mathbb{N}_0$, $C(0)=0$, and $C(t)$ is non-decreasing and right-continuous.

The sensor is part of the system model in Fig.~\ref{fig:system}. Depending on a defined trigger, time or event, the sensor is sampled and an update message that contains the current event count $C(t)$ is sent. The update messages are indexed $n \in \mathbb{N}$ and we denote $A(n)$ and $D(n)$ their arrival time to the network and departure time from the network, respectively. For convenience, we define $A(0) = D(0) = 0$, $A(\nu,n) = A(n) - A(\nu)$, and $D(\nu,n) = D(n) - D(\nu)$ for $n \ge \nu \ge 0$. Generally for all $n \ge 1$ it holds that $D(n) \ge A(n)$ for causality.

In a time-triggered system, update messages are sent by the sensor at times $A(n) = n w$ for $n \ge 1$ where $w \in \mathbb{R}_+$ is the width of the update interval. In an event-triggered system, update messages are sent whenever the number of events since the last update exceeds a threshold $\alpha \in \mathbb{N}$. This happens at times $A(n) = E(n\alpha)$ for $n \ge 1$. We assume that the monitor does not have any other, a priori knowledge of the random sensor process. In particular, it does not know the distribution nor any moments of the sensor process.

Practical examples of our system range from networked leak or overflow sensors, alert counters and alert aggregation in cloud and network operations, to people counting sensors, e.g., at emergency exits. More general sensor models may include processes $C(t)$ that are not non-decreasing. Examples include Gaussian noise and Wiener processes in~\cite{ayan:valueofinformation, champati:ageofinformationfeedbackcontrol, klugel:aoipenalty,sun:remoteestimation} or Markovian random walks. These may cause additional difficulties when defining a condition on the process $C(t)$ that triggers generation of update messages $A(n)$.
%
%------------------------------------------------------------------------
%
\subsection{Definition of Performance Metrics}
The network delay, respectively, the sojourn time of message $n \ge 1$ can be written as
\begin{equation}
T(n) = D(n) - A(n).
\label{eq:delaydef}
\end{equation}

A common definition of AoI at time $t > D(1)$ is $\Delta(t) = t - \max_{n \ge 1} \{A(n) : D(n) < t \}$. This definition matches~\cite{champati:ageofinformationmaxplus} with the minor difference that we define $\Delta(t)$ as a left-continuous function. Thus, the peak AoI of update $n \ge 1$ follows as
\begin{equation}
\Delta(n) = D(n+1) - A(n).
\label{eq:aoidef}
\end{equation}

Complementary to the AoI that is signal-agnostic, we define a signal-aware deviation-of-information (DoI) metric $\Phi(t) = C(t) - \max_{n \ge 1} \{C(A(n)): D(n) < t\}$ for $t > D(1)$. The DoI is the deviation of the current sensor signal from the latest value received by the monitor. The peak DoI of update $n \ge 1$ is
\begin{equation}
\Phi(n) = C(D(n+1)) - C(A(n)),
\label{eq:doidef}
\end{equation}
that is attained at the departure time of update message $n+1$ when the monitor uses the information of update $n$ for the last time.
%
%------------------------------------------------------------------------
%
\section{Delay and AoI Statistics}
\label{sec:aoi}
In this section, we define the queueing model and its statistical characterization. We derive a lemma for delay and AoI that is key to our later analysis of the DoI. This lemma also provides statistical delay $T_{\varepsilon}$ and AoI bounds $\Delta_{\varepsilon}$ that satisfy $\mathsf{P}[T(n) > T_{\varepsilon}] \le \varepsilon$ and $\mathsf{P}[\Delta(n) > \Delta_{\varepsilon}] \le \varepsilon$, respectively.
%
%------------------------------------------------------------------------
%
\subsection{Queueing Model}
We model queueing systems and networks thereof using a definition of a max-plus server~\cite[Def. 1]{fidler:multiserver} that is adapted from the definition of g-server from~\cite[Def. 6.3.1]{chang:performanceguarantees}.
\begin{definition}[Max-Plus Server]
\label{def:maxplusserver}
A system with arrival process $A(n)$ and departure process $D(n)$ is a max-plus server with service process $S(\nu,n)$ if it holds for all $n \ge 1$ that
\begin{equation*}
D(n) \le \max_{\nu \in [1,n]} \{ A(\nu) + S(\nu,n) \}.
\end{equation*}
\end{definition}
The general class of work-conserving, lossless, first-in first-out (fifo) queueing systems satisfies the definition of max-plus server with service process $S(\nu,n) = \sum_{m=\nu}^n L(m)$ where $L(m) \in \mathbb{R}_{+}$ is the service time of message $m \ge 1$~\cite[Lem. 1]{fidler:multiserver}. This includes G$\mid$G$\mid$1 queues~\cite[Ex. 6.2.3]{chang:performanceguarantees}. Since any tandem of max-plus servers is a max-plus server, too, the model extends naturally to networks of queues.

By insertion of Def.~\ref{def:maxplusserver} into the definition of network delay~\eqref{eq:delaydef} it follows readily for $n \ge 1$ that
\begin{equation}
T(n) \le \max_{\nu \in [1,n]} \{ S(\nu,n) - A(\nu,n) \}.
\label{eq:delay}
\end{equation}
Similarly, for the peak AoI~\eqref{eq:aoidef} we obtain for $n \ge 1$ that
\begin{align}
\Delta(n) \le \max \biggl\{ & \max_{\nu \in [1,n]} \{ S(\nu,n+1) - A(\nu,n) \}, \nonumber \\
& S(n+1,n+1) + A(n,n+1) \biggr\} .
\label{eq:aoi}
\end{align}
%
%------------------------------------------------------------------------
%
\subsection{Statistical Characterization}
We derive statistical tail bounds using Chernoff's theorem
\begin{equation}
\mathsf{P}[X \ge x] \le e^{-\theta x} \mathsf{M}_X(\theta),
\label{eq:chernoff}
\end{equation}
for any $\theta > 0$, where $\mathsf{M}_X(\theta) = \mathsf{E}[e^{\theta X}]$ is the moment generating function (MGF) of the random variable $X$ and $x$ is an arbitrary threshold parameter. We will frequently use that $\mathsf{M}_{X+Y}(\theta) = \mathsf{M}_{X}(\theta) \mathsf{M}_{Y}(\theta)$ for statistically independent random variables $X$ and $Y$.

We characterize the MGF of arrival and service processes by $(\sigma,\rho)$-envelopes defined in~\cite[Def. 7.2.1]{chang:performanceguarantees}. These are adapted to max-plus servers in~\cite[Def. 2]{fidler:multiserver}. We use arrival processes with independent and identically distributed (iid) increments $A(n-1,n)$ for $n \ge 1$, including deterministic increments as a special case. For iid increments the parameter $\sigma_A = 0$ and the arrival process is characterized by an envelope rate $\rho_A > 0$.
\begin{definition}[Service and Arrival Envelopes]
\label{def:enveloperates}
Each of the following statements for all $n \ge \nu \ge 1$ and $\theta > 0$. A service process, $S(\nu,n)$, has $(\overline{\sigma}_S(\theta),\overline{\rho}_S(\theta))$-upper envelope if
\begin{equation*}
\mathsf{E}\Bigl[e^{\theta S(\nu,n)}\Bigr] \le e^{\theta (\overline{\sigma}_S(\theta) + \overline{\rho}_S(\theta) (n-\nu+1))} .
\end{equation*}
An arrival process, $A(\nu,n)$, has $\underline{\rho}_A(\theta)$-lower envelope if
\begin{equation*}
\mathsf{E}\Bigl[e^{-\theta A(\nu,n)}\Bigr] \le e^{-\theta \underline{\rho}_A(-\theta) (n-\nu)} ,
\end{equation*}
and $\overline{\rho}_A(\theta)$-upper envelope if
\begin{equation*}
\mathsf{E}\Bigl[e^{\theta A(\nu,n)}\Bigr] \le e^{\theta \overline{\rho}_A(\theta) (n-\nu)} .
\end{equation*}
\end{definition}
Next, we obtain bounds of the MGF of delay and AoI that are an essential building block of the following derivations.
\begin{lemma}[MGF bounds of delay and AoI]
\label{lem:delayaoi}
Given arrivals $A(n)$ with iid increments and envelope parameters $(\underline{\rho}_A, \overline{\rho}_A)$ at a max-plus server $S(\nu,n)$ with envelope parameters $(\overline{\sigma}_S,\overline{\rho}_S)$. For any $\theta > 0$ that satisfies $\underline{\rho}_A(-\theta) > \overline{\rho}_S(\theta)$ it holds for the MGF of the delay $T(n)$ for any $n \ge 1$ that
\begin{equation*}
\mathsf{M}_{T}(\theta) \le \frac{e^{\theta (\overline{\sigma}_S(\theta)+\overline{\rho}_S(\theta))}}{1-e^{-\theta(\underline{\rho}_A(-\theta)-\overline{\rho}_S(\theta))}} ,
\end{equation*}
and for the MGF of the AoI $\Delta(n)$ for any $n \ge 1$ that
\begin{equation*}
\mathsf{M}_{\Delta}(\theta) \le \frac{e^{\theta (\overline{\sigma}_S(\theta)+2\overline{\rho}_S(\theta))}}{1-e^{-\theta(\underline{\rho}_A(-\theta)-\overline{\rho}_S(\theta))}} + e^{\theta (\overline{\sigma}_S(\theta)+\overline{\rho}_S(\theta)+\overline{\rho}_A(\theta))} .
\end{equation*}
\end{lemma}
\begin{proof}
We first show the derivation of the MGF of the delay. The MGF of the AoI follows similarly.
\paragraph{Delay}
We estimate the MGF of the sojourn time using the approach from~\cite{chang:performanceguarantees, fidler:momentcalculus}. It follows from~\eqref{eq:delay} for $n \ge 1$ and $\theta > 0$ that
\begin{align*}
\mathsf{M}_{T}(\theta,n) \le & \mathsf{E}\bigl[e^{\theta \max_{\nu \in [1,n]} \{ S(\nu,n) - A(\nu,n) \}}\bigr] \\
= & \mathsf{E}\biggl[\max_{\nu \in [1,n]} \bigl\{ e^{\theta ( S(\nu,n) - A(\nu,n) )}\bigr\}\biggr] \\
\le & \mathsf{E} \Biggl[ \sum_{\nu=1}^n e^{\theta ( S(\nu,n) - A(\nu,n) )} \Biggr] \\
= & \sum_{\nu=1}^n \mathsf{E} \bigl[ e^{\theta S(\nu,n)} \bigr] \mathsf{E} \bigl[e^{-\theta A(\nu,n)} \bigr] ,
\end{align*}
where we used independence of $S(\nu,n)$ and $A(\nu,n)$. By insertion of the envelope parameters we have
\begin{align*}
\mathsf{M}_{T}(\theta,n) \le & e^{\theta (\overline{\sigma}_S(\theta)+\overline{\rho}_S(\theta))} \sum_{\nu=1}^n \Bigl(e^{-\theta(\underline{\rho}_A(-\theta)-\overline{\rho}_S(\theta))}\Bigr)^{n-\nu} \\
\le & e^{\theta (\overline{\sigma}_S(\theta)+\overline{\rho}_S(\theta))} \sum_{\nu=0}^{\infty} \Bigl(e^{-\theta(\underline{\rho}_A(-\theta)-\overline{\rho}_S(\theta))}\Bigr)^{\nu} ,
\end{align*}
where $\sum_{\nu=0}^{\infty} x^\nu = 1/(1-x)$ if $x < 1$ concludes the proof, implying the stability condition $\underline{\rho}_A(-\theta) > \overline{\rho}_S(\theta)$.
\paragraph{AoI}
We use the same essential steps to estimate the MGF of the AoI. From~\eqref{eq:aoi} we have for $n \ge 1$ and $\theta > 0$ that
\begin{align*}
\mathsf{M}_{\Delta}(\theta,n)
%\le & \mathsf{E} \Bigl[ e^{\theta \max \{ \max_{\nu \in [1,n]} \{ S(\nu,n+1) - A(\nu,n) \}, S(n+1,n+1) + A(n,n+1) \}} \Bigr] \\
\le & \sum_{\nu=1}^{n} \mathsf{E} \bigl[ e^{\theta S(\nu,n+1)} \bigr] \mathsf{E} \bigl[e^{-\theta A(\nu,n) } \bigr] \\
+ & \mathsf{E} \bigl[ e^{\theta S(n+1,n+1)} \bigr] \mathsf{E} \bigl[e^{\theta A(n,n+1 )} \bigr] \\
\le & e^{\theta (\overline{\sigma}_S(\theta)+2\overline{\rho}_S(\theta))} \sum_{\nu=1}^n \Bigl(e^{-\theta(\underline{\rho}_A(-\theta)-\overline{\rho}_S(\theta))}\Bigr)^{n-\nu} \\
+ & e^{\theta (\overline{\sigma}_S(\theta)+\overline{\rho}_S(\theta)+\overline{\rho}_A(\theta))} .
\end{align*}
Again, $\underline{\rho}_A(-\theta) > \overline{\rho}_S(\theta)$ achieves convergence if $n \rightarrow \infty$.
\end{proof}
%
%------------------------------------------------------------------------
%
\subsection{Statistical Performance Bounds}
\begin{figure}
\centering
\includegraphics[width=0.66\linewidth]{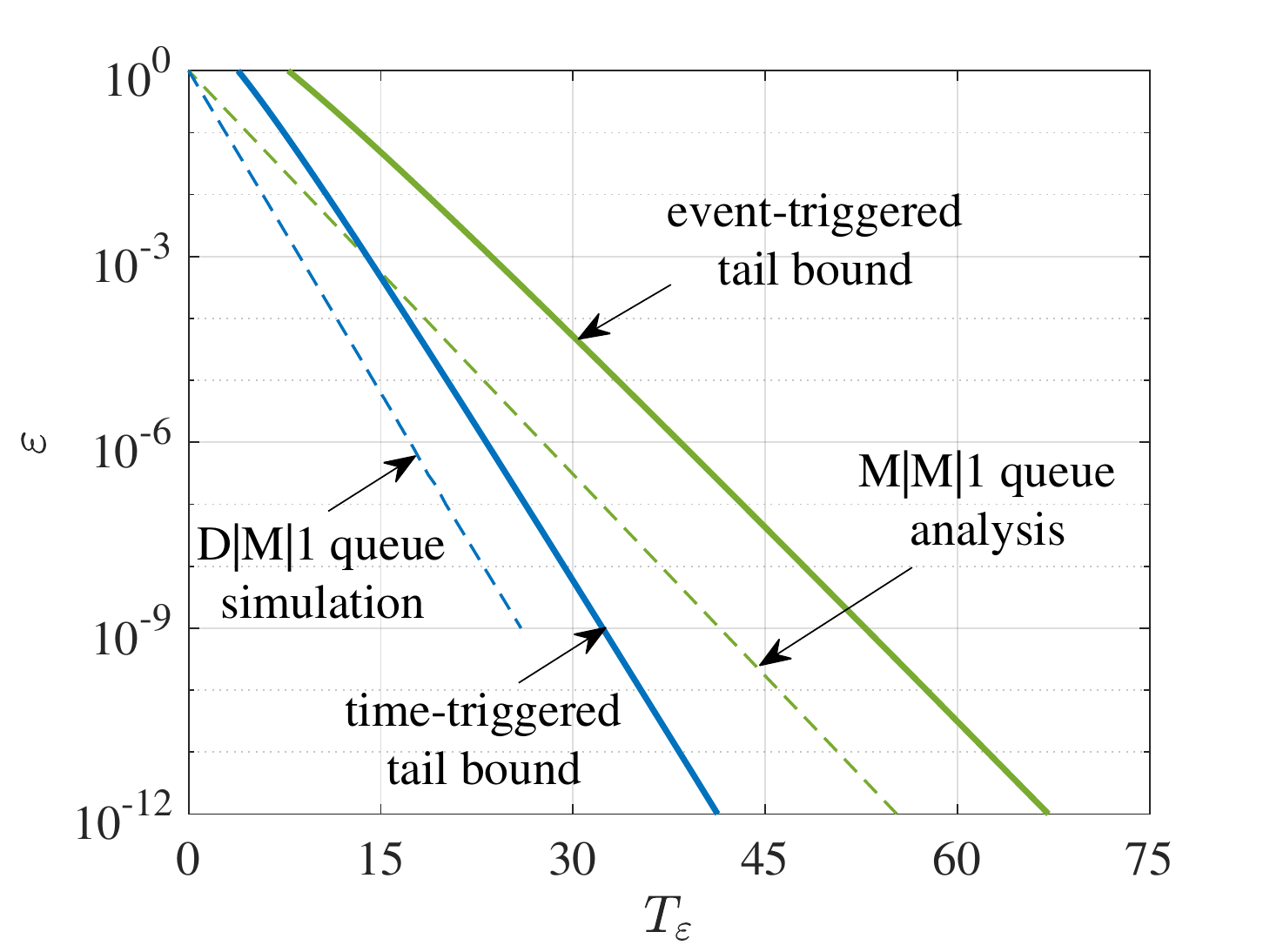}
\caption{Sojourn time bounds of the time-triggered system and the event-triggered system with exponential inter-event times, exponential service times, and parameter $\alpha = 1$. In this case, the time-triggered system is a D$\mid$M$\mid$1 queue and corresponding simulation results are shown for comparison, and the event-triggered system is an M$\mid$M$\mid$1 queue that has a known tail distribution.}
\label{fig:dm1mm1}
\end{figure}
Statistical delay and AoI bounds follow as an immediate corollary of Lem.~\ref{lem:delayaoi} and Chernoff's theorem~\eqref{eq:chernoff}. Specifically, we have for the delay for any $n \ge 1$ and $\theta > 0$ that
\begin{equation*}
\mathsf{P}[T(n) \ge T_{\varepsilon}] \le e^{-\theta T_{\varepsilon}} \mathsf{M}_T(\theta) =: \varepsilon .
\end{equation*}
Solving for $T_{\varepsilon}$ we have that
\begin{equation}
T_{\varepsilon}(\theta) = \frac{\ln \mathsf{M}_T(\theta) - \ln \varepsilon}{\theta} ,
\label{eq:statisticaldelaybound}
\end{equation}
and similarly for the AoI
\begin{equation}
\Delta_{\varepsilon}(\theta) = \frac{\ln \mathsf{M}_\Delta(\theta) - \ln \varepsilon}{\theta} ,
\label{eq:statisticalaoibound}
\end{equation}
are statistical upper bounds of delay and AoI, respectively, that are exceeded at most with probability $\varepsilon$. Since $T_{\varepsilon}(\theta)$ and $\Delta_{\varepsilon}(\theta)$ are valid upper bounds for any $\theta > 0$, we can optimize $\theta > 0$ to find the smallest upper bounds. Next, we evaluate these bounds for time-triggered and event-triggered systems, respectively.
%
%------------------------------------------------------------------------
%
\subsubsection{Time-triggered systems}
For a time-triggered system where update messages are generated at times $A(n) = n w$ for $n \ge 1$ and $w \in \mathbb{R}_+$ is the width of the update interval, the envelope parameters in Def.~\ref{def:enveloperates} for all $\theta > 0$ are simply
\begin{equation}
\underline{\rho}_A = w, \quad\quad \overline{\rho}_A = w .
\label{eq:timetriggeredparameters}
\end{equation}
%
%------------------------------------------------------------------------
%
\subsubsection{Event-triggered systems}
\begin{figure*}
\subfigure[$\lambda = 0.25$]{
\includegraphics[width=0.32\linewidth]{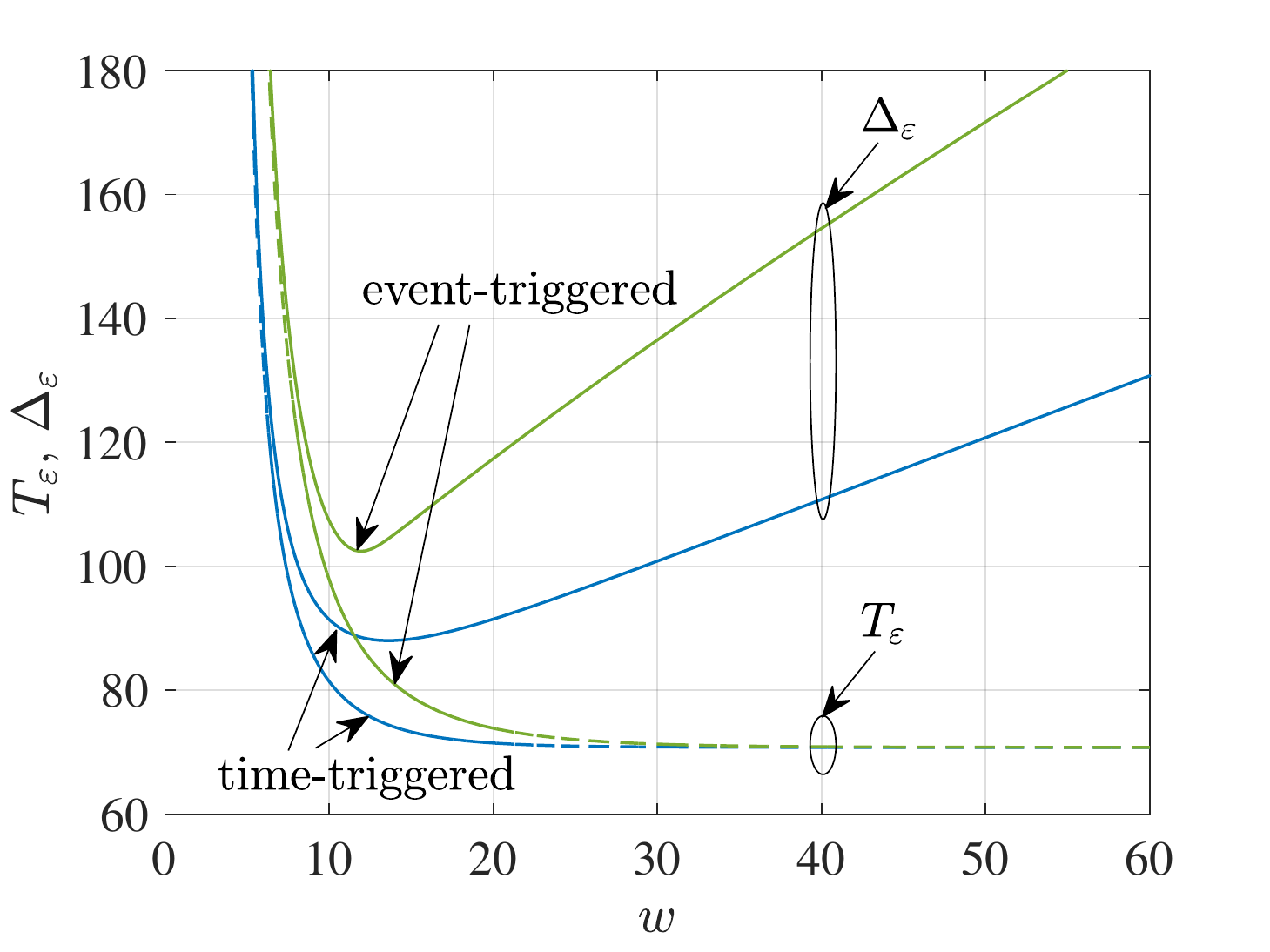}
\label{fig:delayaoilambda025}
}
\subfigure[$\lambda = 0.5$]{
\includegraphics[width=0.32\linewidth]{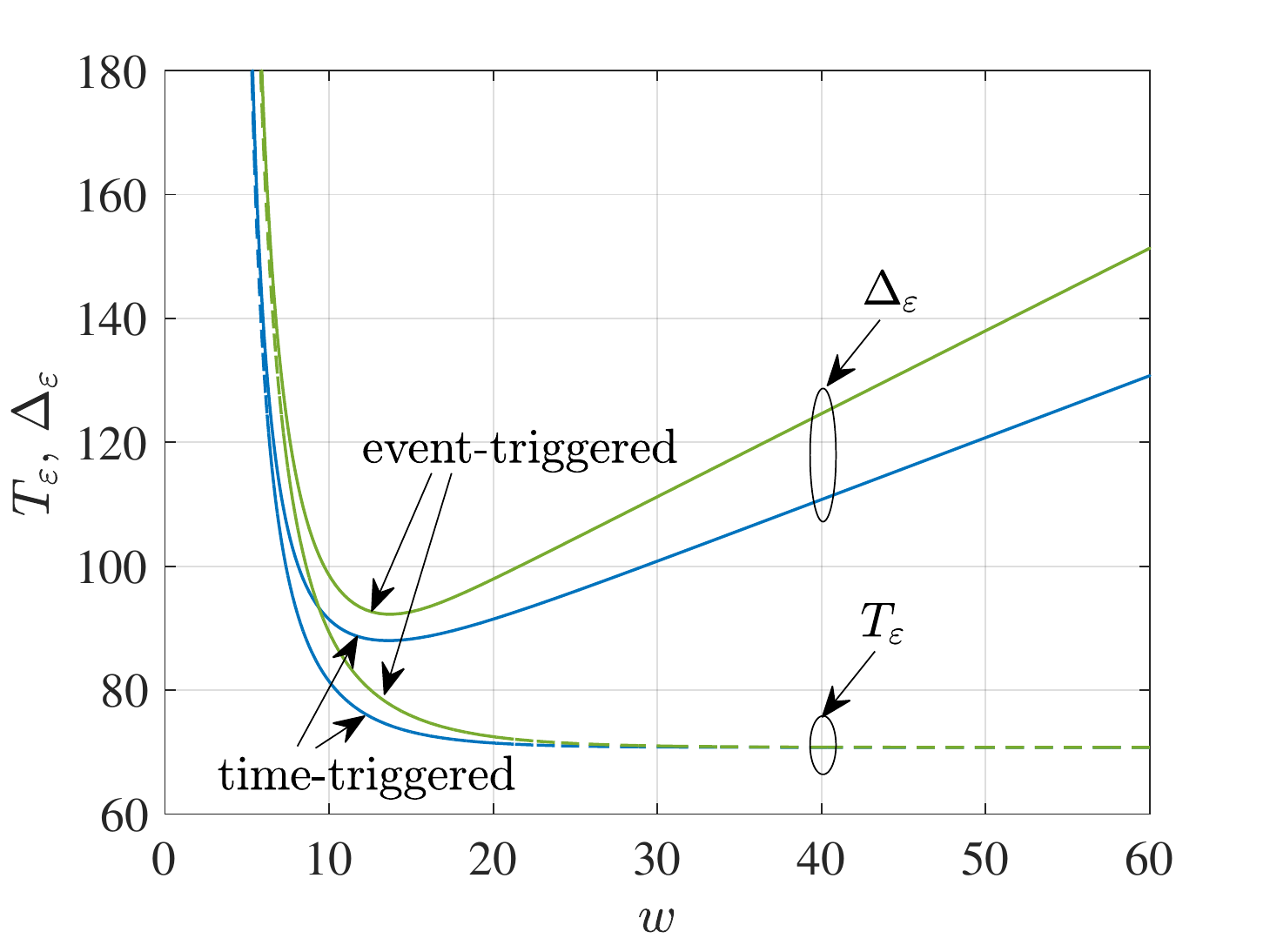}
\label{fig:delayaoilambda05}
}
\subfigure[$\lambda = 1$]{
\includegraphics[width=0.32\linewidth]{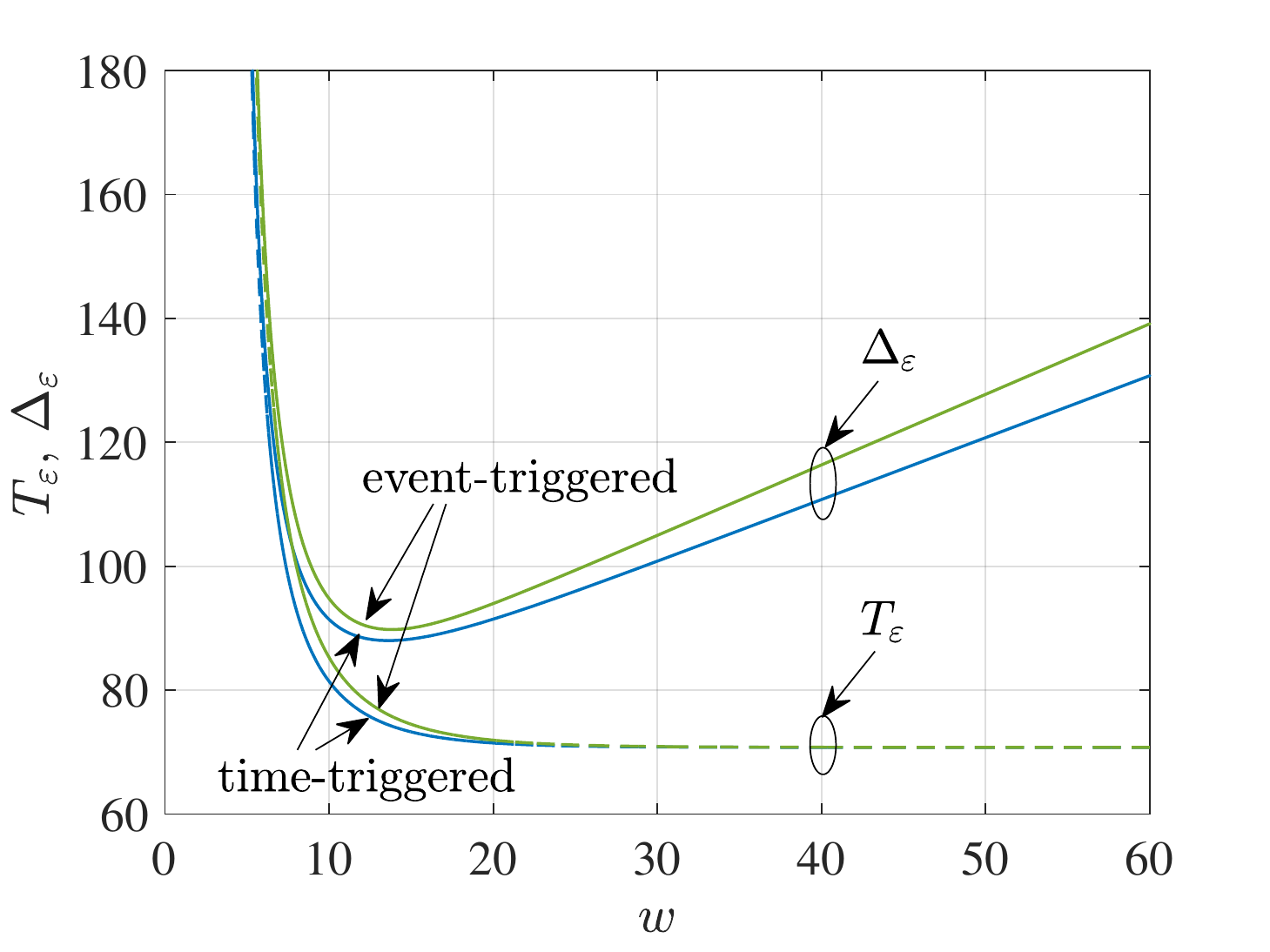}
\label{fig:delayaoilambda1}
}
\caption{Sojourn time and AoI bounds for $\varepsilon = 10^{-6}$ for the time-triggered and the event-triggered system. Inter-event times are exponential with parameter $\lambda$. The update interval $w$ of the time-triggered system and the event threshold $\alpha$ of the event-triggered system are varied, where $\alpha = \lambda w$ achieves the same mean network utilization for both systems.}
\label{fig:delayaoilambda}
\end{figure*}
For an event-triggered system, $A(n) = E(n\alpha)$ for $n \ge 1$ and $\alpha \in \mathbb{N}$ is a threshold parameter. We assume that inter-event times $I(n)$ are iid with MGF $\mathsf{M}_I(\theta)$. With $A(n) = \sum_{\nu=1}^{n\alpha} I(\nu)$, it follows for $\theta > 0$ that
\begin{equation}
\underline{\rho}_A(-\theta) = -\frac{\alpha}{\theta} \ln (\mathsf{M}_I(-\theta)), \quad \overline{\rho}_A(\theta) = \frac{\alpha}{\theta} \ln (\mathsf{M}_I(\theta)).
\label{eq:eventtriggeredparameters}
\end{equation}

If the time between events is exponential with parameter $\lambda > 0$ we have for $\theta < \lambda$ that
\begin{equation*}
\mathsf{M}_{I}(\theta) = \frac{\lambda}{\lambda-\theta} .
\end{equation*}
In this case, the sensor signal $C(t)$ is a Poisson counting process with parameter $\lambda$. Further, the time between two event-triggered update messages is iid Erlang with $\alpha$ and $\lambda$.
%
%------------------------------------------------------------------------
%
\subsubsection{Service times}
We consider messages of variable length and denote $L(n)$ the service time of message $n \ge 1$. It holds that $S(\nu,n) = \sum_{m=\nu}^n L(m)$~\cite[Lem. 1]{fidler:multiserver} and considering iid service times it follows for $\theta > 0$ that $\overline{\sigma}_S = 0$ and
\begin{equation}
\overline{\rho}_S(\theta) = \frac{1}{\theta} \ln (\mathsf{M}_{L}(\theta)) .
\label{eq:exposerviceparameter}
\end{equation}
Considering exponential service times with parameter $\mu > 0$ we have for $\theta < \mu$ that
\begin{equation*}
\mathsf{M}_{L}(\theta) = \frac{\mu}{\mu-\theta} .
\end{equation*}
We will also consider the case of deterministic message service times $L(n)=l$ for $n \ge 1$ and $l > 0$ which gives $\overline{\rho}_S = l$.
%
%------------------------------------------------------------------------
%
\subsubsection{Numerical results}
Statistical delay and AoI bounds follow from~\eqref{eq:statisticaldelaybound} and~\eqref{eq:statisticalaoibound}, respectively, by insertion of the envelope parameters~\eqref{eq:timetriggeredparameters} or~\eqref{eq:eventtriggeredparameters}, and~\eqref{eq:exposerviceparameter} into Lem.~\ref{lem:delayaoi}. We optimize the free parameter $\theta$ numerically to obtain the smallest upper bound.

The time-triggered system is a D$\mid$G$\mid$1 queue or in case of exponential service times a D$\mid$M$\mid$1 queue, respectively. The event-triggered system is of type G$\mid$G$\mid$1, respectively, Erlang-$\alpha$$\mid$M$\mid$1 in case of exponential inter-event times and exponential service times. For $\alpha=1$ it becomes a basic M$\mid$M$\mid$1 queue. For reference, the exact tail distribution of $T_{\varepsilon}$ of the M$\mid$M$\mid$1 queue is known~\cite{bolch:queueingnetworks} as
\begin{equation}
\varepsilon = e^{-\mu\left(1-\frac{\lambda}{\mu}\right)T_{\varepsilon}} .
\label{eq:exactmm1tail}
\end{equation}

In Fig.~\ref{fig:dm1mm1} we display the tail decay of sojourn time bounds of the time-triggered and the event-triggered system with exponential inter-event times with parameter $\lambda=0.5$ and exponential service times with parameter $\mu=1$. We consider the case $\alpha=1$ for the event-triggered system. For the time-triggered system we choose parameter $w=2$ that achieves the same average network utilization. For comparison, we include empirical quantiles from $10^9$ sojourn time samples obtained by simulation of a D$\mid$M$\mid$1 queue and the tail distribution of the M$\mid$M$\mid$1 queue~\eqref{eq:exactmm1tail}. The tail bounds exhibit the correct speed of tail decay and show the expected accuracy~\cite{fidler:netcalcguide}.

In Fig.~\ref{fig:delayaoilambda} we compare delay and AoI bounds with probability $\varepsilon = 10^{-6}$ of the time-triggered and the event-triggered system. Service times and inter-event times are exponential, where the service rate is $\mu=0.25$ and different sensor event rates $\lambda \in \{0.25, 0.5, 1\}$ are used. While the arrival process of the time-triggered system is not affected by $\lambda$, the arrival process of the event-triggered system is Erlang with parameters $\alpha$ and $\lambda$. We show results for different update intervals $w$ and we set the event threshold $\alpha = \lambda w$, that is the mean number of events during an interval of duration $w$, to achieve the same average utilization for the time-triggered and the event-triggered system.

\begin{figure}
\centering
\vspace{-16pt}
\includegraphics[width=0.66\linewidth]{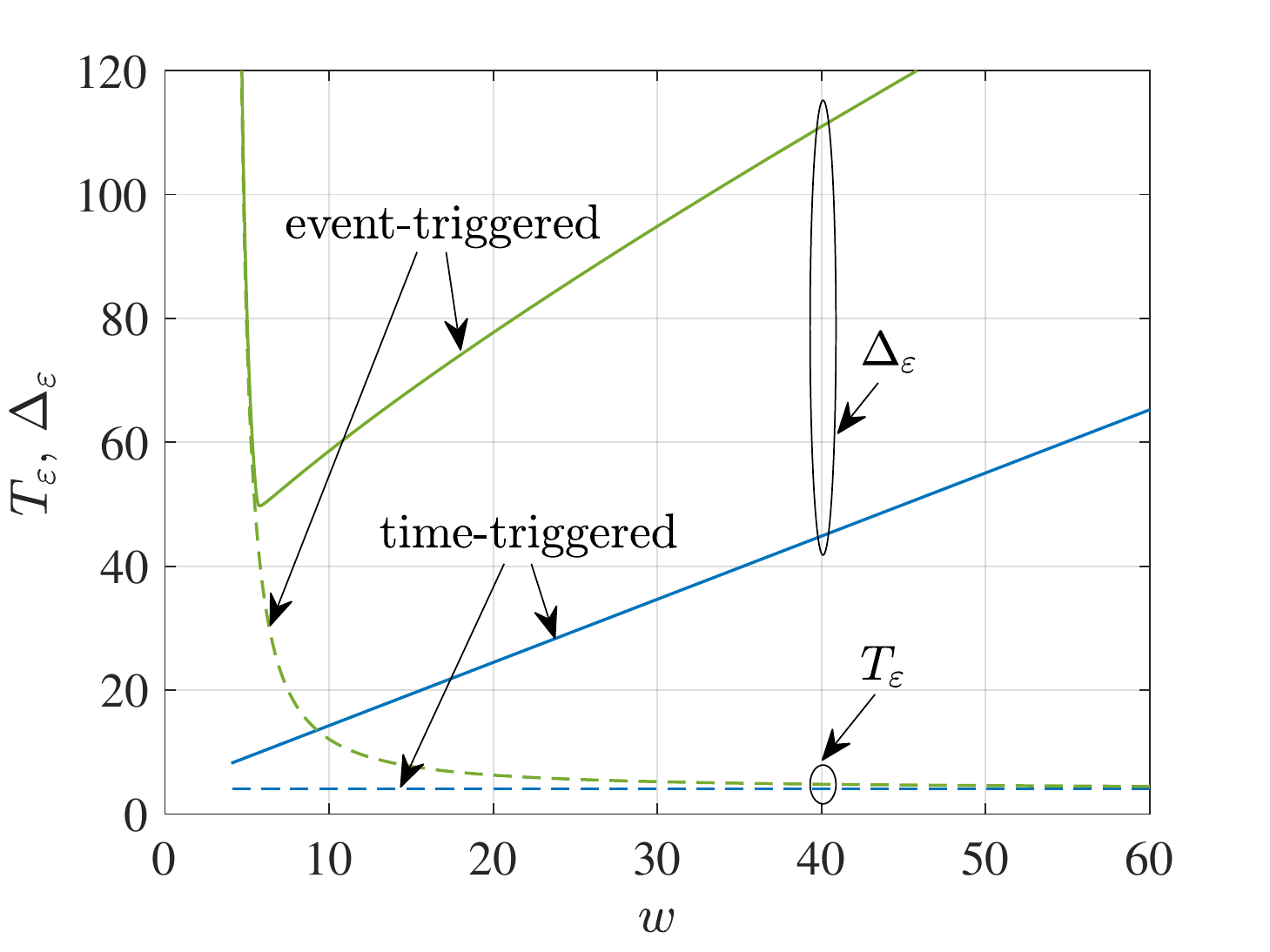}
\caption{Same parameters as in Fig.~\ref{fig:delayaoilambda05} but deterministic service times.}
\vspace{-12pt}
\label{fig:delayaoilambda05detchannel}
\end{figure}

It can be observed that all curves in Fig.~\ref{fig:delayaoilambda} show a tremendous increase if $w$ and $\alpha$ become small. This corresponds to high network utilization that induces queueing delays. In case of large $w$ and $\alpha$, the network delay converges to the service time quantile of a message, whereas the AoI grows almost linearly due to increasingly rare update messages. Generally, it can be observed that the event-triggered system shows worse delay and AoI performance than the time-triggered system. Similar observations have also been made for periodic versus random arrivals in~\cite{kaul:ageofinformationqueue, inoue:aoisingleserverqueues, champati:ageofinformationgigiqueue, modiano:informationfreshness}. This is a consequence of the variability of the arrival process of the event-triggered system that leads to two different effects: bursts of update messages cause queueing delays in the network, this effect is dominant in the left of the graphs in Fig.~\ref{fig:delayaoilambda}; or the absence of update messages causes idle waiting, dominant in the right of the graphs. With increasing $\lambda$ and $\alpha$ the arrival process becomes smoother and the performance of the event-triggered system approaches that of the time-triggered system, see Fig.~\ref{fig:delayaoilambda1}.

Fig.~\ref{fig:delayaoilambda05detchannel} uses the same parameters as Fig.~\ref{fig:delayaoilambda05} with the exception that the network service times are deterministic, i.e., the queue is served with a constant service rate of $0.25$. In this case, the time-triggered system is a D$\mid$D$\mid$1 queue and the bounds obtained from Lem.~\ref{lem:delayaoi} correctly identify the delay $T_{\varepsilon} = 4$ and the AoI $\Delta_{\varepsilon} = 4+w$ for all $w > 4$. The event-triggered system is an Erlang-$\alpha$$\mid$D$\mid$1 queue. For small $\alpha$ corresponding to high utilization the burstiness of the arrivals causes large queueing delays. With increasing $\alpha$ the queueing delays diminish quickly and the system switches sharply to a regime, where the AoI is dominated by idle waiting due to too infrequent update messages.
%
%------------------------------------------------------------------------
%
\section{DoI Bounds}
\label{sec:doi}
In this section, we investigate how event-triggered systems perform compared to time-triggered systems if we consider the signal-aware DoI metric. We derive statistical bounds of the DoI of time-triggered and event-triggered systems and show numerical as well as simulation results.
%
%------------------------------------------------------------------------
%
\subsection{Analysis}
We derive statistical bounds of the peak DoI $\Phi_{\varepsilon}$ that satisfy $\mathsf{P}[\Phi(n) > \Phi_{\varepsilon}] \le \varepsilon$. The analysis of DoI is more involved due to the use of the doubly stochastic processes $C(A(n))$ and $C(D(n))$. As before, we consider time-triggered systems, where update messages are generated at times $A(n) = n w$ for $n \ge 1$ and $w \in \mathbb{R}_+$ is the width of the update interval, and event-triggered systems, where update messages are generated at times $A(n) = E(n \alpha)$ and $\alpha \in \mathbb{N}$ is the event threshold, respectively. The following theorem uses Lem.~\ref{lem:delayaoi} to state our main result.
\begin{theorem}[DoI bounds]
\label{th:doibounds}
Given the assumptions of Lem.~\ref{lem:delayaoi}. Consider events with iid inter-event times $I(n)$ for $n \ge 1$ and denote $J(t)$ the residual inter-event time at time $t \ge 0$.

For the DoI $\Phi(n)$ of a time-triggered system with update interval $w$ and envelope parameters~\eqref{eq:timetriggeredparameters}, it holds for all $n \ge 1$, $\theta > 0$, and $\Phi_{\varepsilon} \in \mathbb{N}_0$ that
\begin{equation*}
\mathsf{P}[\Phi(n) > \Phi_{\varepsilon}] \le \mathsf{M}_{\Delta}(\theta) \mathsf{M}_{J(A(n))}(-\theta) (\mathsf{M}_I(-\theta))^{\Phi_{\varepsilon}} .
\end{equation*}

For the DoI $\Phi(n)$ of an event-triggered system with threshold $\alpha$, and envelope parameters~\eqref{eq:eventtriggeredparameters}, it holds for all $n \ge 1$, $\theta > 0$, and $\Phi_{\varepsilon} \in \mathbb{N}_0 \ge \alpha-1$ that
\begin{equation*}
\mathsf{P}[\Phi(n) > \Phi_{\varepsilon}] \le \mathsf{M}_T(\theta) (\mathsf{M}_I(-\theta))^{\Phi_{\varepsilon}-\alpha+1} .
\end{equation*}
\end{theorem}
The MGF of the residual inter-event time can be estimated as $\mathsf{M}_{J(t)}(-\theta) \le 1$ for $\theta > 0$. For a memoryless distribution we also have $\mathsf{M}_{J(t)}(-\theta) = \mathsf{M}_{I}(-\theta)$.

Equating the bound for time-triggered systems in Th.~\ref{th:doibounds} with $\varepsilon$ and considering a memoryless inter-event distribution, we can solve for
\begin{equation*}
\Phi_{\varepsilon} = \biggl\lceil\frac{\ln \varepsilon - \ln \mathsf{M}_{\Delta}(\theta)}{\ln \mathsf{M}_I(-\theta)} \biggl\rceil - 1,
\end{equation*}
and for event-triggered systems
\begin{equation*}
\Phi_{\varepsilon} = \biggl\lceil\frac{\ln \varepsilon - \ln \mathsf{M}_{T}(\theta)}{\ln \mathsf{M}_I(-\theta)} \biggl\rceil  + \alpha - 1.
\end{equation*}
\begin{proof} We start with the proof for event-triggered systems, since time-triggered systems pose some additional difficulties.
%
%------------------------------------------------------------------------
%
\paragraph{Event-triggered system}
By definition of the event-triggered system we have $C(A(n)) = n \alpha$. Using~\eqref{eq:eventcount}, we also have $C(D(n+1)) = \max \{\nu \ge 0: E(\nu) \le D(n+1)\}$. Further, for the last expression we know that $\nu \ge (n+1)\alpha$, since $D(n+1) \ge A(n+1)$ and hence $C(D(n+1)) \ge C(A(n+1)) = (n+1)\alpha$. By insertion into~\eqref{eq:doidef} it holds for $n \ge 0$ that
\begin{equation*}
\Phi(n) = \max \{\nu \ge (n+1)\alpha: E(\nu) \le D(n+1)\} - n \alpha .
\end{equation*}
With a variable substitution it follows that
\begin{equation*}
\Phi(n) = \alpha + \max \{\nu \ge 0: E((n+1)\alpha+\nu) \le D(n+1)\}.
\end{equation*}
We use $D(n+1) = A(n+1) + T(n+1)$ and $A(n+1) = E((n+1)\alpha) = \sum_{m=1}^{(n+1)\alpha} I(m)$ to obtain
\begin{equation*}
\Phi(n) = \alpha + \max \Biggl\{\nu \ge 0: \sum_{m=(n+1)\alpha+1}^{(n+1)\alpha+\nu} I(m) \le T(n+1)\Biggr\}.
\end{equation*}
Now, choose some $\Phi_{\varepsilon} \in \mathbb{N}_0 \ge \alpha-1$. The case $\Phi(n) > \Phi_{\varepsilon}$ occurs iff $\nu = \Phi_{\varepsilon}-\alpha+1$ satisfies the condition above, i.e., $\sum_{m=(n+1)\alpha+1}^{(n+1)\alpha+\nu} I(m) \le T(n+1)$. It follows that
\begin{equation*}
\mathsf{P}[\Phi(n) > \Phi_{\varepsilon}] = \mathsf{P} \Biggl[T(n+1) - \sum_{m=(n+1)\alpha+1}^{(n+1)\alpha+\Phi_{\varepsilon}-\alpha+1} I(m) \ge 0  \Biggr] .
\end{equation*}
With Chernoff's theorem~\eqref{eq:chernoff} we have $\mathsf{P}[X \ge 0] \le \mathsf{M}_X(\theta)$ for $\theta > 0$ so that
\begin{equation*}
\mathsf{P}[\Phi(n) > \Phi_{\varepsilon}] \le \mathsf{M} \Biggl[T(n+1) - \sum_{m=(n+1)\alpha+1}^{(n+1)\alpha+\Phi_{\varepsilon}-\alpha+1} I(m) \Biggr] (\theta) .
\end{equation*}
The result of Th.~\ref{th:doibounds} follows for iid inter-event times $I(m)$. Note that for iid inter-event times $T(n+1)$ is independent of events that occur after $A(n+1) = E((n+1)\alpha)$.
%
%------------------------------------------------------------------------
%
\begin{figure*}
\subfigure[Deterministic events, exponential service]{
\includegraphics[width=0.32\linewidth]{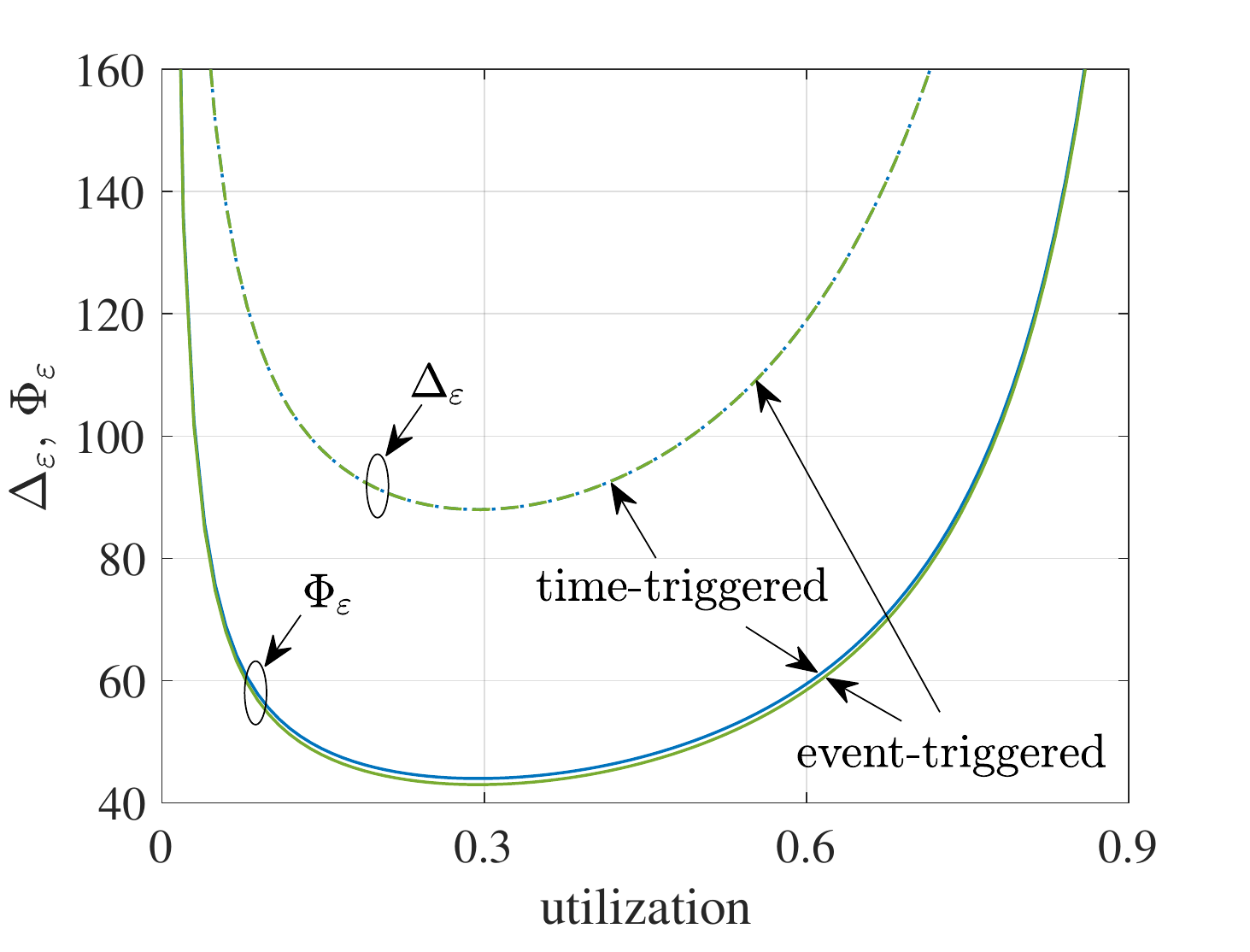}
\label{fig:delayaoidoidetsensor}
}
\subfigure[Exponential events, exponential service]{
\includegraphics[width=0.32\linewidth]{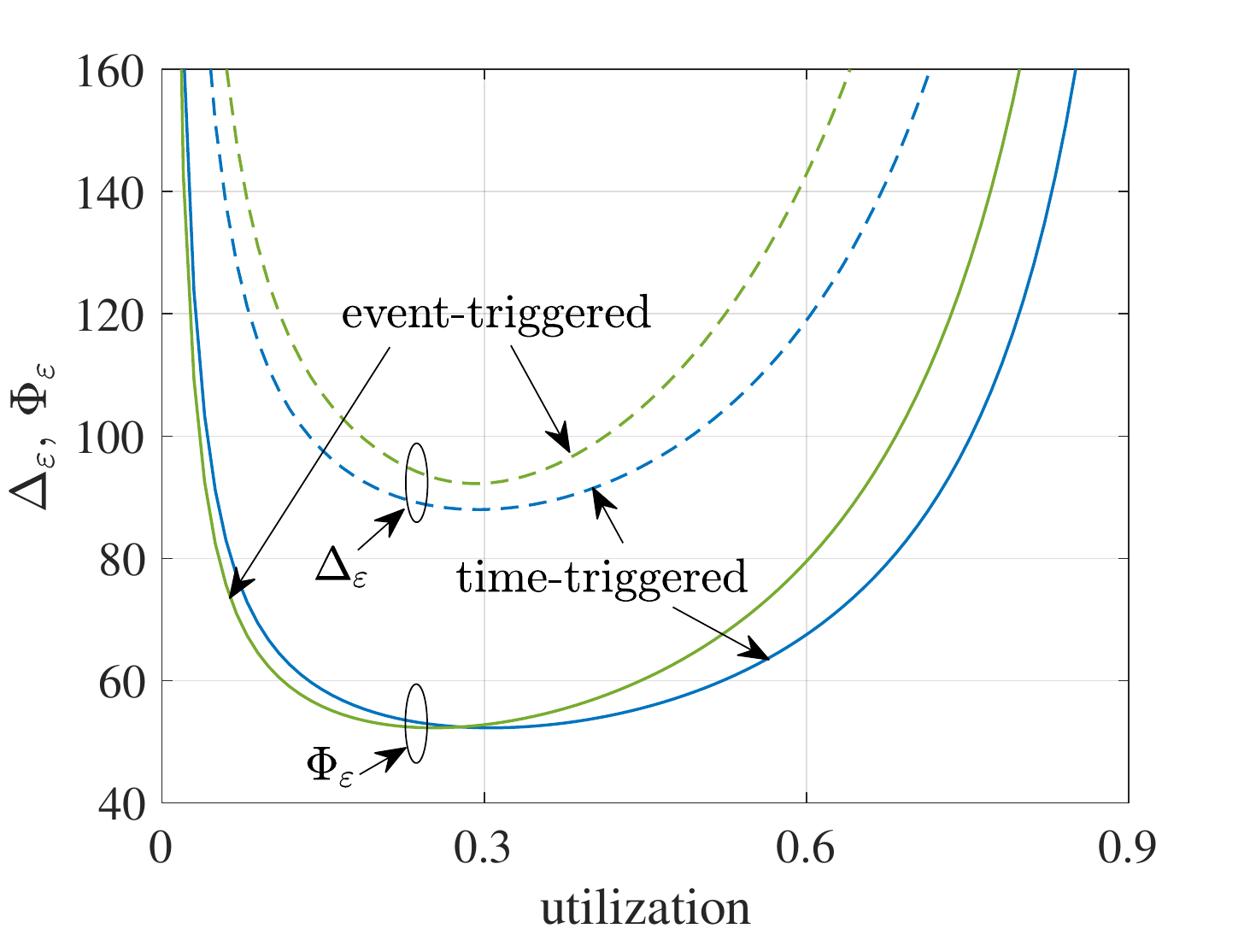}
\label{fig:doiexpchannel}
}
\subfigure[Exponential events, deterministic service]{
\includegraphics[width=0.32\linewidth]{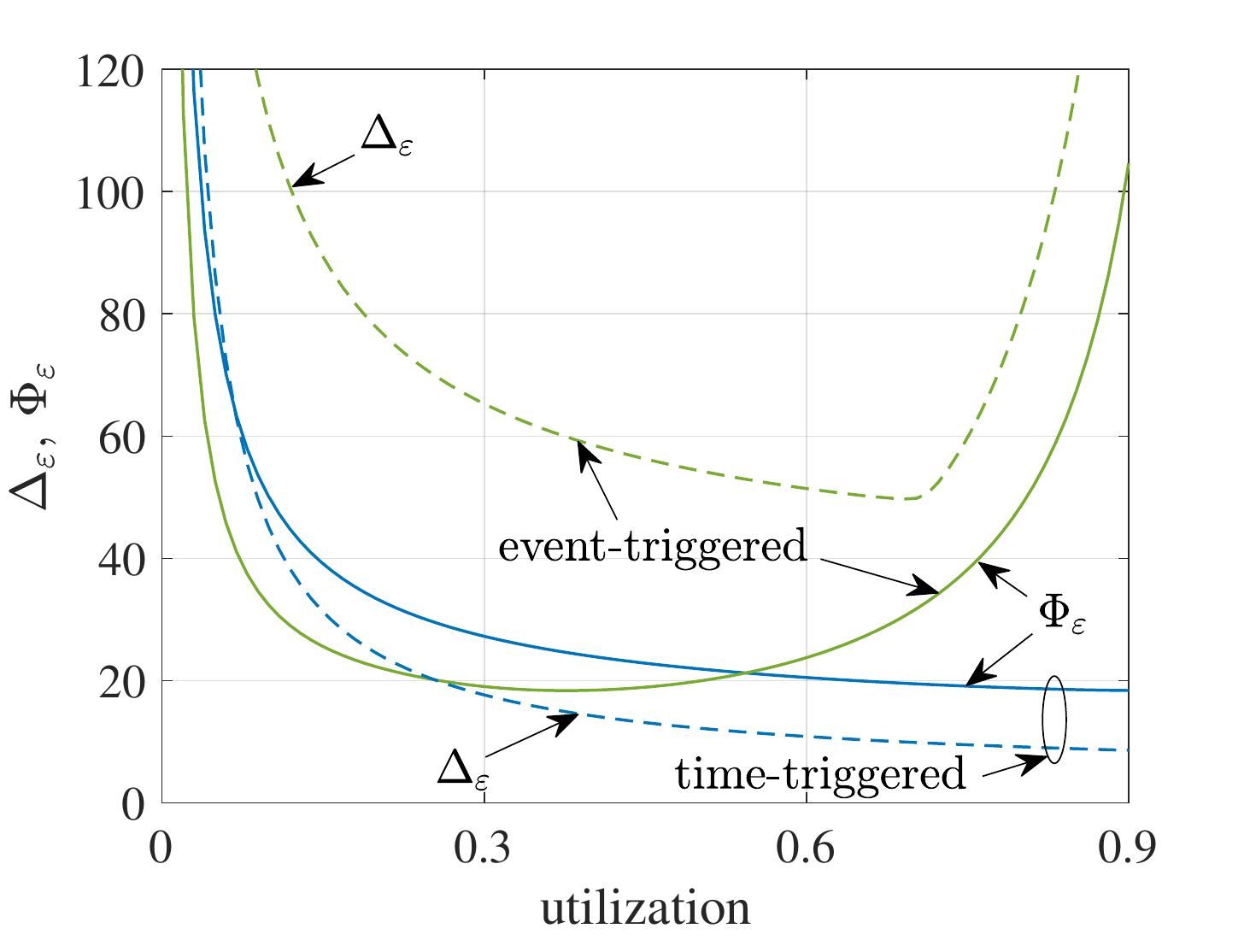}
\label{fig:doidetchannel}
}
\caption{AoI and DoI bounds for $\varepsilon = 10^{-6}$ for the time-triggered and the event-triggered system. The update interval width $w$ and the event threshold $\alpha$ are adjusted to achieve the desired network utilization.}
\label{fig:doi}
\end{figure*}
\paragraph{Time-triggered systems}
For time-triggered systems, we have the additional difficulty that the generation of messages is not synchronized with the occurrence of events. Instead, at time $t \ge 0$, e.g., $t=A(n)$, we only know that the last event occurred at time $E(C(t))$ and the next event occurs at time $E(C(t)+1) = E(C(t))+I(C(t)+1)$. We denote $J(t)$ the residual inter-event time at time $t \ge 0$ until the next event occurs, i.e., $J(t) = E(C(t)+1) - t$. It follows that
\begin{equation}
J(t) = E(C(t)) + I(C(t)+1) - t .
\label{eq:residualtime}
\end{equation}

First, we formalize an intermediate result. Consider some times $t,\tau \ge 0$. From~\eqref{eq:eventcount} we have
\begin{equation}
C(t+\tau) = \max \{\nu \ge C(t): E(\nu) \le t + \tau \} .
\label{eq:intermediate1}
\end{equation}
For $\nu \ge C(t) + 1$ we can write
\begin{align}
E(\nu) = & E(C(t)) + \sum_{m=C(t)+1}^{\nu} I(m) \nonumber \\
= & t +  J(t) + \sum_{m=C(t)+2}^{\nu} I(m),
\label{eq:intermediate2}
\end{align}
where we use~\eqref{eq:residualtime} in the second step. By insertion of~\eqref{eq:intermediate2} for $\nu \ge C(t)+1$ into~\eqref{eq:intermediate1} and noting that the case $\nu = C(t)$ is trivial, we obtain that
\begin{align*}
& C(t+\tau) \\
= & \max \Biggl\{ \nu \ge C(t): J(t)1_{\nu \ge C(t)+1} + \sum_{m=C(t)+2}^{\nu} I(m) \le \tau \Biggr\} \\
= & C(t) + \max \Biggl\{\nu \ge 0 : J(t)1_{\nu \ge 1} + \sum_{m=C(t)+2}^{C(t)+\nu} I(m) \le \tau \Biggr\},
\end{align*}
where $1_{(.)}$ is the indicator function that is one if the argument is true and zero otherwise.

Next, we insert $D(n+1) = A(n) + \Delta(n)$ from~\eqref{eq:aoidef} into~\eqref{eq:doidef} and with the previous result we obtain by substitution of $t=A(n)$ and $\tau = \Delta(n)$ for $n \ge 1$ that
\begin{multline*}
\Phi(n) = C(A(n)+\Delta(n)) - C(A(n)) = \\
\max \Biggl\{ \nu \ge 0 : J(A(n))1_{\nu \ge 1} + \sum_{m=C(A(n))+2}^{C(A(n))+\nu} I(m) \le \Delta(n) \Biggr\}.
\end{multline*}
Now, choose some $\Phi_{\varepsilon} \in \mathbb{N}_0$. The case $\Phi(n) > \Phi_{\varepsilon}$ occurs iff $\nu = \Phi_{\varepsilon}+1$ satisfies the condition above. It follows that
\begin{multline*}
\mathsf{P}[\Phi(n) > \Phi_{\varepsilon}] \\ = \mathsf{P}\Biggl[ \Delta(n) - J(A(n)) - \sum_{m=C(A(n))+2}^{C(A(n))+\Phi_{\varepsilon}+1} I(m) \ge 0 \Biggr] .
\end{multline*}
With Chernoff's theorem~\eqref{eq:chernoff} we have for $\theta > 0$ that
\begin{multline*}
\mathsf{P}[\Phi(n) > \Phi_{\varepsilon}] \\ \le \mathsf{M}\Biggl[ \Delta(n) - J(A(n)) - \sum_{m=C(A(n))+2}^{C(A(n))+\Phi_{\varepsilon}+1} I(m) \Biggr](\theta) .
\end{multline*}
The result of Th.~\ref{th:doibounds} follows for iid inter-event times $J(A(n))$ and $I(m)$. We note that in a time-triggered system $\Delta(n)$ is independent of the occurrence of events.
\end{proof}
%
%------------------------------------------------------------------------
%
\subsection{Numerical Results}
\begin{figure*}
\subfigure[Empirical distribution]{
\includegraphics[width=0.32\linewidth]{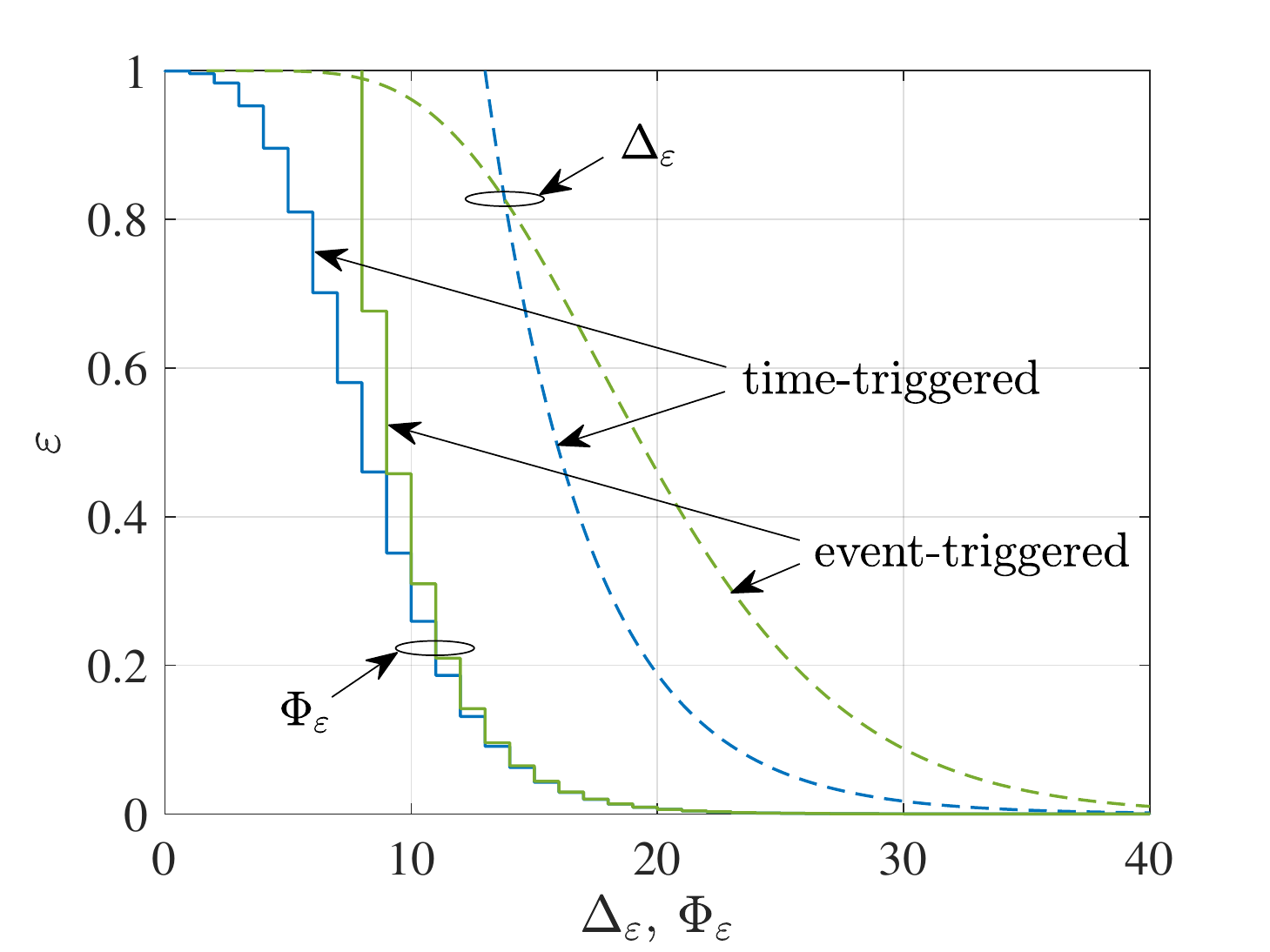}
\label{fig:aoidoilinearsim05}
}
\subfigure[Tail decay, time-triggered system]{
\includegraphics[width=0.32\linewidth]{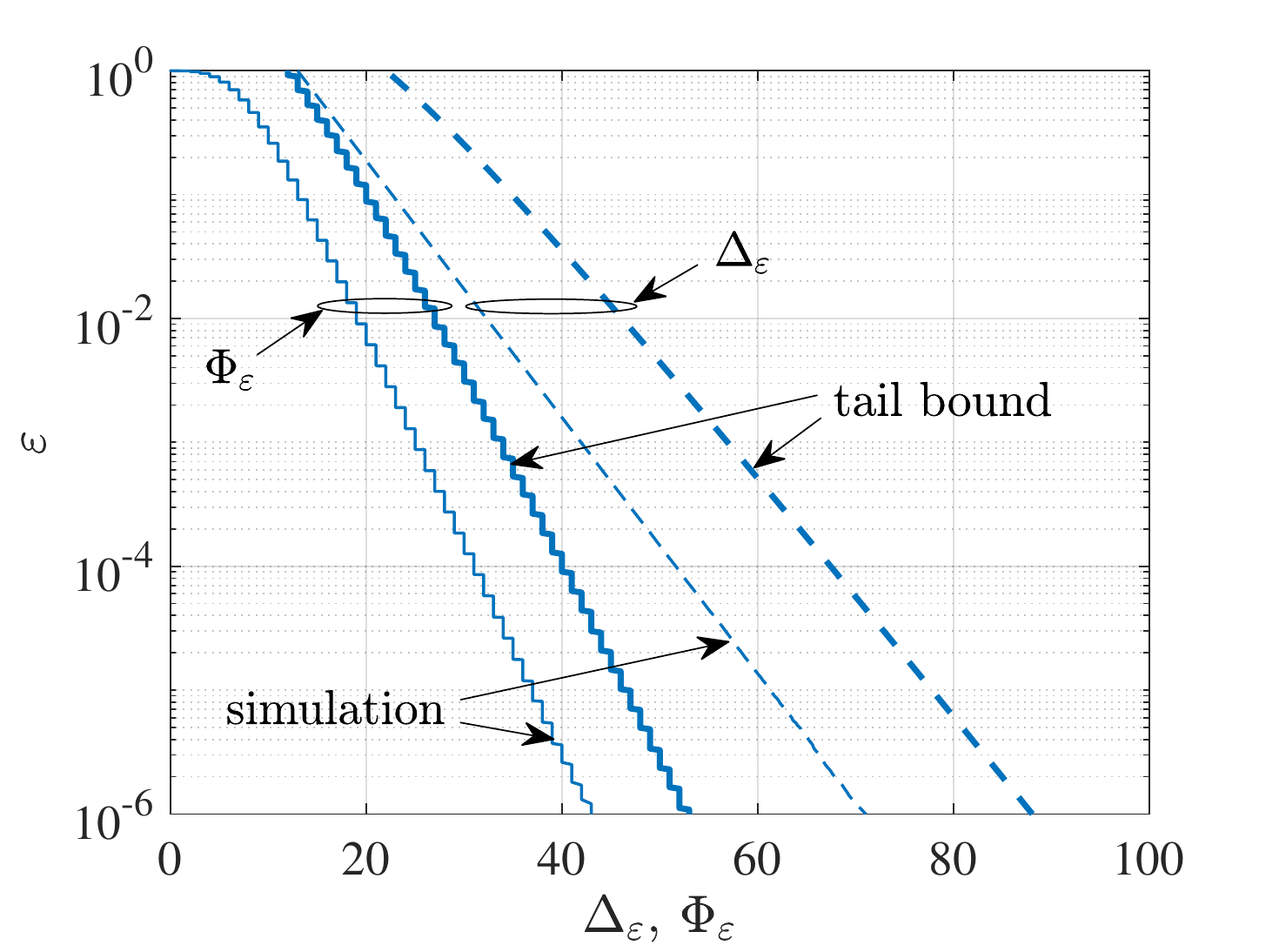}
\label{fig:aoidoittsim05}
}
\subfigure[Tail decay, event-triggered system]{
\includegraphics[width=0.32\linewidth]{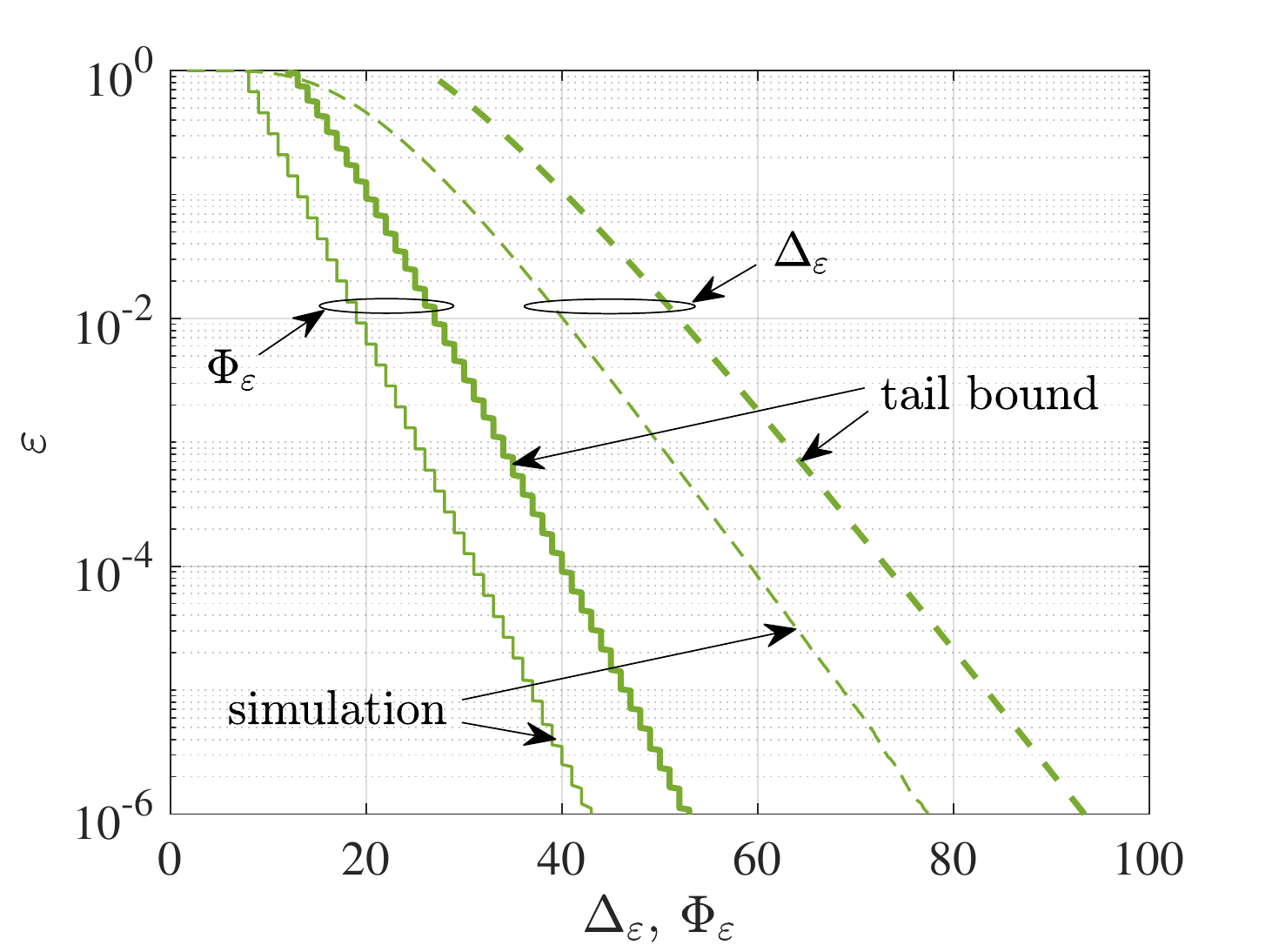}
\label{fig:aoidoietsim05}
}
\caption{AoI and DoI distribution for the systems in Fig.~\ref{fig:doiexpchannel}. For the time triggered-systems the update interval width is $w=13$ and for the event-triggered system the event threshold is $\alpha=8$. These parameters corresponds to a utilization of $0.325$ and $0.25$, respectively, that minimize the $\varepsilon = 10^{-6}$ DoI bounds.}
\label{fig:doisim}
\end{figure*}
In Fig.~\ref{fig:doi} we show tail bounds of the AoI and DoI for $\varepsilon = 10^{-6}$. The bounds are derived using Lem.~\ref{lem:delayaoi} and Th.~\ref{th:doibounds}. The free parameter $\theta$ is optimized numerically. We consider a range of relevant time-triggered and event-triggered systems. In all cases, the mean rate of sensor events is $\lambda = 0.5$ and the mean service rate of the network queue is $\mu =0.25$. The width of the update interval $w$ of the time-triggered system and the event threshold $\alpha$ of the event-triggered system are varied in unison so that both cause the same network utilization, that is $1/(w\mu)$ and $\lambda/(\alpha\mu)$, respectively. We use the network utilization as the abscissa. For reasons of presentability, we mostly ignore the integer constraints of $\alpha$ and $\Delta_{\varepsilon}$ in the figures.
\paragraph{Deterministic events, exponential service}
In Fig.~\ref{fig:delayaoidoidetsensor} we consider exponential network service times with parameter $\mu$ and a deterministic sensor signal, i.e., periodic events with deterministic inter-event times $1/\lambda = 2$. This degenerate case serves as a reference. In this case both, the time-triggered and the event-triggered system, sent updates periodically. We choose $\alpha = \lambda w$ to ensure the same network utilization resulting in identical delay and AoI bounds.

The DoI bounds differ slightly since update messages are synchronized with the occurrence of sensor events in the event-triggered system but not in the time-triggered system. This is reflected by the residual inter-event time $J(t)$ in Th.~\ref{th:doibounds}. Since deterministic inter-event times are not memoryless, we estimate $\mathsf{M}_{J(t)}(-\theta) < 1$ for $\theta > 0$ by $1$ and obtain with Th.~\ref{th:doibounds} for the time-triggered system that
\begin{equation*}
\Phi_{\varepsilon} = \frac{\ln \varepsilon - \ln \mathsf{M}_{\Delta}(\theta)}{\ln \mathsf{M}_I(-\theta)}  = \frac{\lambda (\ln \mathsf{M}_{\Delta}(\theta) - \ln \varepsilon)}{\theta} = \lambda \Delta_{\varepsilon},
\end{equation*}
where we inserted the MGF $\mathsf{M}_I(-\theta) = e^{-\theta/\lambda}$ of the deterministic inter-event time $1/\lambda$ and ignored integer constraints. In the final step, we substituted the AoI bound $\Delta_{\varepsilon}$~\eqref{eq:statisticalaoibound}. This implies that the update rate that achieves the minimal AoI also minimizes the DoI in this case. As can be observed in Fig.~\ref{fig:delayaoidoidetsensor}, the minimal AoI bound $\Delta_{\varepsilon} = 88$ and the minimal DoI bound $\Phi_{\varepsilon} = 44$, corresponding to $\lambda = 0.5$, are achieved for the same network utilization of about 0.3.
\paragraph{Exponential events, exponential service}
The direct correspondence of AoI and DoI $\Phi_{\varepsilon} = \lambda \Delta_{\varepsilon}$ observed in Fig.~\ref{fig:delayaoidoidetsensor} is, however, not given in case of a random sensor signal. In Fig.~\ref{fig:doiexpchannel} we show results for exponential instead of deterministic inter-event times. All other parameters are unchanged. The same set of parameters has also been used for Fig~\ref{fig:delayaoilambda05}.

For the time-triggered system, that is signal-agnostic, the AoI is generally unaffected by the choice of the sensor model. Consequently, the AoI in Fig.~\ref{fig:doiexpchannel} is identical to Fig.~\ref{fig:delayaoidoidetsensor}. The DoI increases, however, since a varying number of sensor events may occur during any update interval.

In case of the event-triggered system, the AoI in Fig.~\ref{fig:doiexpchannel} is larger than in Fig.~\ref{fig:delayaoidoidetsensor} since the arrivals to the network are now a random process. Due to the randomness, the AoI of the event-triggered system is generally larger than the AoI of the time-triggered system, as also observed in Fig.~\ref{fig:delayaoilambda}.

Regarding the DoI, the event-triggered system has the advantage that it is signal-aware and sends update messages only if needed. Interestingly, both systems, time-triggered and event-triggered, show comparable minimal DoI. For an intuitive explanation consider a burst of sensor events. In this case, the event-triggered system samples the sensor more frequently with the goal to improve the DoI. The increased rate of update messages may, however, cause network congestion and queueing delays that are detrimental to the DoI and outweigh their advantage. Overall this appears to cause similar minimal DoI, however, at a lower average network utilization for the event-triggered system. Concluding, the u-shaped DoI curves in Fig.~\ref{fig:doiexpchannel} show that both systems are feasible and robust to variations of the network utilization. Configured optimally, the event-triggered system uses less network resources. It generates, however, more bursty network traffic.

A related finding in~\cite{champati:ageofinformationfeedbackcontrol, klugel:aoipenalty} is that the problem of minimizing the mean-square norm of the state error at the monitor is equivalent to a signal-agnostic AoI minimization problem. In case of our event-triggered and hence signal-aware system, Fig.~\ref{fig:doiexpchannel} does not confirm a similar result. Here, the network utilization that achieves the minimal tail bounds is different for the AoI and DoI, respectively.
\paragraph{Exponential events, deterministic service}
Fig.~\ref{fig:doidetchannel} shows results for the same system as in Fig.~\ref{fig:doiexpchannel} but with deterministic service times $1/\mu = 4$ as also used in Fig.~\ref{fig:delayaoilambda05detchannel}. In this case the time-triggered system is purely deterministic and achieves a very small AoI that is determined as the sum of the network service time and the width of the update interval. Hence, the AoI is minimal in case of full network utilization. The same applies for the DoI bound.

The event-triggered system shows a much larger AoI that is due to the randomness of the update messages. For low utilization, corresponding to a large threshold $\alpha$, the AoI is large due to infrequent updates if the sensor signal does not change much. In case of high utilization, small $\alpha$, queueing delays start to dominate and the AoI bends sharply upwards.

Despite the large AoI, the event-triggered system achieves a similarly good minimal DoI bound as the time-triggered system. Specifically at low utilization, the DoI bound of the event-triggered system is much smaller. This is a consequence of the deterministic network service, where the delivery of an update message within one message service time $1/\mu = 4$ is almost guaranteed, given the utilization is low and queueing delays are avoided. This is particularly favorable for the event-triggered system since once the sensor signal changes by more than the threshold $\alpha$, an update message can be delivered with high probability within short time.
\paragraph{Decay of tail probabilities}
In Fig.~\ref{fig:doiexpchannel} the minimal DoI bound of the time-triggered system is achieved for $w \approx 13$ and of the event-triggered system for $\alpha = 8$, corresponding to utilizations of $0.325$ and $0.25$, respectively. We investigate these parameters in more detail in Fig.~\ref{fig:doisim} where we show AoI and DoI bounds as well as empirical quantiles from $10^8$ samples of the AoI and DoI obtained by simulation. While the minimal DoI bounds of the time-triggered and event-triggered systems in Fig.~\ref{fig:doiexpchannel} are about the same for $\varepsilon=10^{-6}$, we see in Fig.~\ref{fig:aoidoilinearsim05} that the DoI quantiles differ if $\varepsilon$ is not small. Particularly, the DoI approaches $\alpha$ for $\varepsilon \rightarrow 1$ in case of the event-triggered system and zero in case of the time-triggered system. Conversely, the AoI approaches $w$ for $\varepsilon \rightarrow 1$ in case of the time-triggered system and zero in case of the event-triggered system. We do not display tail bounds for the range of $\varepsilon$ in Fig.~\ref{fig:aoidoilinearsim05}. We include the bounds in Fig.~\ref{fig:aoidoittsim05} and Fig.~\ref{fig:aoidoietsim05} where we show the tail decay. It can be noticed that the DoI bounds and the empirical DoI quantiles of the time-triggered and the event-triggered system exhibit the same speed of tail decay. This dominates the DoI if $\varepsilon$ is small causing similar DoI performance for both systems.

In Fig.~\ref{fig:doisimnonopt} we include simulation results for non-optimal parameters $w$ and $\alpha$. For smaller $w$ and $\alpha$ we see an improvement of the AoI and DoI if $\varepsilon$ is not small. This is due to more frequent update messages. At the same time this causes increased network utilization and a smaller speed of tail decay. This consumes the initial advantage when $\varepsilon$ becomes small and leads to worse tail performance. In case of larger than optimal $w$ and $\alpha$, update messages are sent less frequently so that the AoI and DoI increase. This also brings about a reduction of the network utilization that can, however, only achieve a small improvement of the speed of the tail decay which is not relevant for $\varepsilon = 10^{-6}$.
\begin{figure}
\subfigure[Time-triggered system]{
\includegraphics[width=0.46\linewidth]{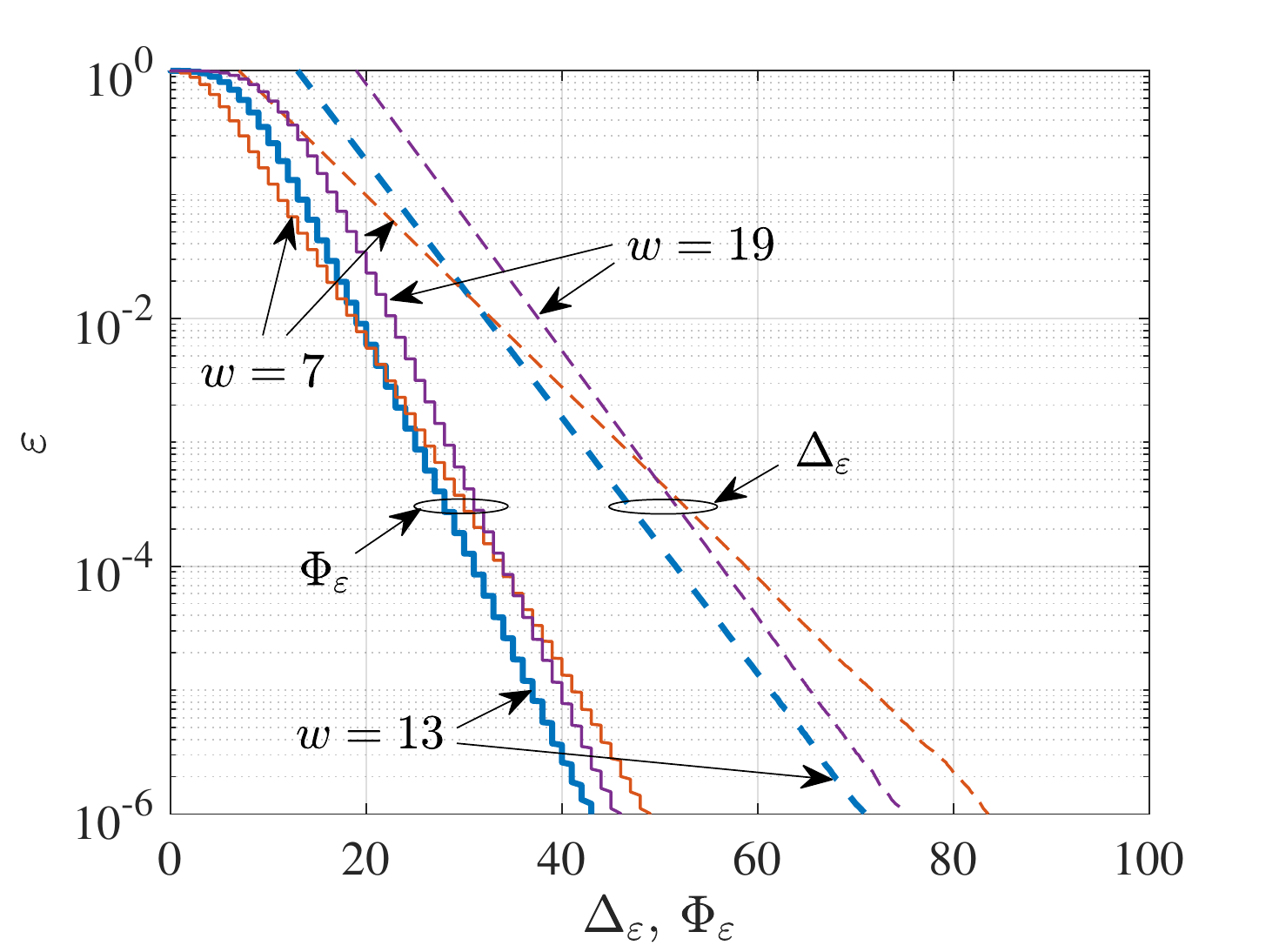}
\label{fig:aoidoittsim05w_7_13_19}
}
\subfigure[Event-triggered system]{
\includegraphics[width=0.46\linewidth]{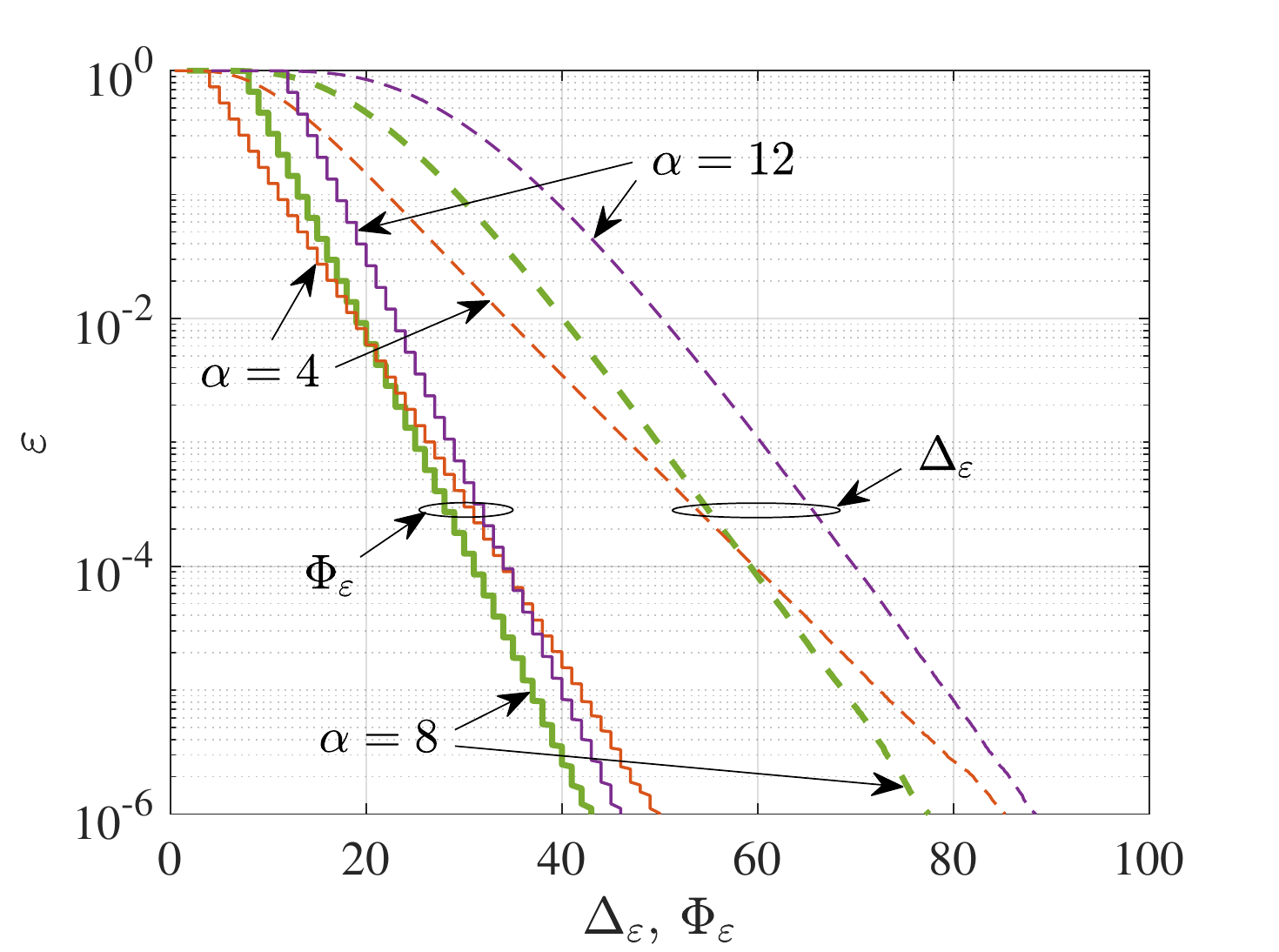}
\label{fig:aoidoietsimalpha_4_8_12}
}
\caption{Empirical AoI and DoI distribution for the system in Fig.~\ref{fig:doiexpchannel}. Parameters $w=13$ and $\alpha=8$ minimize the $\varepsilon=10^{-6}$ DoI bound. It can be seen how smaller or larger parameters are sub-optimal for $\varepsilon=10^{-6}$.}
\label{fig:doisimnonopt}
\end{figure}
%
%------------------------------------------------------------------------
%
\section{Conclusions}
\label{sec:conclusion}
We considered remote monitoring of a sensor via a network. The sampling policy of the sensor is either time-triggered or event-triggered. Correspondingly, sampling is either signal-agnostic or signal-aware. We derived tail bounds of the delay and peak age-of-information that show advantages of the time-triggered system. These metrics do, however, not take the estimation error at the monitor into account, motivating a complementary definition of deviation-of-information. Despite inferior age-, we find that the event-triggered system achieves similar deviation-of-information as the time-triggered system. Sending update messages only in case of certain sensor events, the event-triggered system operates optimally at a lower network utilization and saves network resources.
%
%------------------------------------------------------------------------
%
%\balance
\bibliographystyle{IEEEtran}
\bibliography{IEEEabrv,IEEEfidler}
%
%------------------------------------------------------------------------
%
\end{document}